\documentclass[sigconf]{acmart}
\AtBeginDocument{%
  \providecommand\BibTeX{{%
    \normalfont B\kern-0.5em{\scshape i\kern-0.25em b}\kern-0.8em\TeX}}}

\usepackage{balance}
\usepackage{bm}
\usepackage[normalem]{ulem}
\usepackage{csquotes}
\usepackage[capitalize]{cleveref}
\usepackage{amsmath,amsfonts}
\usepackage{array}
\usepackage{comment}
\usepackage{float}
\usepackage{url}
\usepackage{tikz}
\usepackage{hyperref}
\usepackage{soul}
\usepackage{siunitx}
\usepackage[utf8]{inputenc}
\usepackage{caption}
\usepackage{fancyvrb}
\usepackage[linesnumbered,ruled,vlined]{algorithm2e}
\usepackage{algorithmic}
\SetKwInput{KwInput}{Input}
\SetKwInput{KwOutput}{Output}
\SetKwInput{KwAssumption}{Assumption}
\usepackage{lipsum}                     
\usepackage{xargs}                      
\usepackage{xcolor}  

\usepackage[normalem]{ulem}
\def\MinHer(#1,#2){\operatorname{Min}_{\operatorname{Her}}(#1, #2)}

\usepackage[prependcaption,textsize=tiny]{todonotes}
\newcommandx{\unsure}[2][1=]{\todo[linecolor=red,backgroundcolor=red!25,
    bordercolor=red,#1]{#2}}
\newcommandx{\change}[2][1=]{\todo[linecolor=blue,backgroundcolor=blue!25,
    bordercolor=blue,#1]{#2}}
\newcommandx{\info}[2][1=]{\todo[linecolor=green,backgroundcolor=green!25,
    bordercolor=green,#1]{#2}}
\newcommandx{\improvement}[2][1=]{\todo[linecolor=Plum,backgroundcolor=Plum!25,
    bordercolor=Plum,#1]{#2}}
\newcommandx{\thiswillnotshow}[2][1=]{\todo[disable,#1]{#2}}

\usepackage{enumitem,kantlipsum}

\newcommand{\K}{\mathbb K}
\newcommand{\Q}{\mathbb Q}

\newcommand{\disc}{\operatorname{disc}}

\newcommand{\bF}{\mathbf F}
\newcommand{\mN}{\mathcal{N}}

\newcommand{\sing}{\operatorname{\sing}}

\newcommand{\bfield}{\mathbb{K}}
\def\field#1{\bfield(#1)}
\def\acfield#1{\overline{\field#1}}
\def\card#1#2#3{\sharp_{#1}(#2, #3)}

\def\cfiber#1#2#3{{\mathcal{F}_{#1}(#2, #3)}}
\def\bbeta{\bm{\beta}}
\newcommand{\balpha}{{{\bm{\alpha}}}}
\newcommand{\bmf}{{{\bm{f}}}}

\newtheorem{theo}{Theorem}[section]
\theoremstyle{acmdefinition}
\newtheorem{remark}[theo]{Remark}
\newtheorem{prop}{Proposition}
\newtheorem{fact}{Fact}
\newtheorem{notation}{Notation}

\def\gathen#1{{#1}}

\def\scrW{\mathscr{W}}
\def\stickW{\mathcal{S}_k(\check{z_1}, \mathscr{W})}
\def\stickV{\mathcal{S}_k(\check{z_1}, \mathscr{V})}
\def\scrV{\mathscr{V}}
\AtBeginDocument{%
  \providecommand\BibTeX{{%
    \normalfont B\kern-0.5em{\scshape i\kern-0.25em b}\kern-0.8em\TeX}}}

\usepackage{hyperref}
\hypersetup{
    linktocpage=true,
      hyperfootnotes=false,
      pdftex,                %
      bookmarks         = true,
      bookmarksnumbered = true,
      pdfpagemode       = None,
      pdfstartview      = FitH,
      pdfpagelayout     = SinglePage,
      colorlinks        = true,
      linkcolor= magenta, 
      citecolor         = magenta,
      urlcolor          = blue,
      pdfborder         = {0 0 0}
    }

\renewcommand\footnotetextcopyrightpermission[1]{} 
\copyrightyear{2023}
\acmYear{2023}
\setcopyright{acmlicensed}\acmConference[Submitted]{Submitted}{April 25}{2023}
\acmBooktitle{}
\acmPrice{15.00}
\acmDOI{???.???/???.???}
\acmISBN{???}

\begin{document}

\title{Fast Algorithms for Discrete Differential Equations}

\author{Alin Bostan}
\affiliation{%
  \institution{Inria}
  \city{Palaiseau}
  \country{France}
}
\email{alin.bostan@inria.fr}
\author{Hadrien Notarantonio}
\affiliation{%
  \institution{Inria}
  \city{Palaiseau}
  \country{France}
}
\email{hadrien.notarantonio@inria.fr}

\author{Mohab Safey El Din}
\affiliation{%
  \institution{Sorbonne Université}
  \city{Paris}
  \country{France}
}
\email{mohab.safey@lip6.fr}

\renewcommand{\shortauthors}{}

\begin{abstract}
Discrete Differential Equations (DDEs) are functional equations that 
relate algebraically a power series $F(t,u)$ in~$t$ 
with polynomial coefficients in a ``catalytic'' variable~$u$
and the specializations, say at~$u=1$, of $F(t,u)$ and of some of its partial derivatives
in~$u$.
DDEs occur frequently in combinatorics, especially in map enumeration.
If a DDE is of fixed-point type then its solution $F(t,u)$ is unique, 
and a general result by Popescu (1986)   
implies that $F(t,u)$  is an {algebraic} power series.
Constructive proofs of algebraicity for solutions of fixed-point type DDEs 
were proposed in 2006 by Bousquet-M\'elou and Jehanne.
Last year, Bostan et al. initiated a systematic algorithmic study of such DDEs of order~1.
We generalize this study to DDEs of arbitrary order. 
First, we propose nontrivial extensions of algorithms based on polynomial elimination and on the guess-and-prove paradigm.
Second, we design two brand-new algorithms that exploit the special structure of the underlying polynomial systems.
Last, but not least, we report on implementations that are able to solve highly challenging
DDEs with a combinatorial origin. 
\end{abstract}



\maketitle        

\section{Introduction}
\medskip
{\em Context and motivation.}
Enumerative combinatorics contains a vast landscape of nontrivial counting problems
that can hardly be solved without introducing the associated generating 
functions~\cite{Ardila15}. 
Setting up and solving a functional equation for such a generating function 
generally makes it possible to deduce properties of the discrete objects of interest 
(e.g. closed formulas~\cite[p.~104]{Drmota97} and
asymptotic behaviors~\cite[p.~147]{FlSe09}, \cite[\S 5.4.5]{Schaeffer15}).
Among these problems, many require refining the initial enumeration in order
to more easily write a functional equation~\cite{Brown64, Brown65a, Tutte68}. 
Algebraically, this process amounts to introducing
an additional variable, called~\emph{catalytic}~\cite{Zeilberger00}.\\
\indent The so-obtained functional equation
with one catalytic variable usually relates the refined generating function to
its specializations with respect to the catalytic variable, one of these specializations
being the generating function of the initial enumeration problem.
A standard way to deduce combinatorial properties from the functional equation is first to
determine the nature of the generating function~\cite[Chapter~6]{Stanley99}
(e.g. rational, algebraic, D-finite, \ldots),
then to compute a witness (e.g. an annihilating polynomial or an annihilating linear
differential equation, \ldots).\\

\medskip {\em A typical example.} One enumeration problem consists in studying bicolored maps
(black, white) such that the degree of each black face is~$3$ and the degree of each white
face is a multiple of~$3$. Such objects are
called~\emph{$3$-constellations}~\cite[Fig.~6]{BMJ06};
they have been enumerated via a bijective approach in~\cite{BMS00}.
Let us define~$a_n$ as the number of~$3$-constellations with~$n$ black faces.
We consider the refined sequence~$a_{n, d}$ as the number of~$3$-constellations
having~$n$ black faces and outer degree~$3d$.
Let~$F(t, u) = \sum_{n, d\geq 0}{a_{n, d}u^dt^n}\in\mathbb{Q}[u][[t]]$ be its generating
function\footnote{$F(t,u)$ has
polynomial coefficients in~$u$ since for a fixed number of black faces, the outer degree
is finite.}.
The catalytic variable is~$u$ and
the specialization~$F(t, 1)$ is the generating function of the sequence~$(a_n)_{n\geq 0}$.
Using a classical ``deletion of the root edge'' argument~\cite[Fig.~7]{BMJ06}, one gets the
following functional equation with 1 catalytic variable:
\begin{align}\label{eqn:funceqn_3const}
F(t, u) = 1 \;&+ tu(2F(t, u) + F(t, 1))\frac{F(t, u) - F(t, 1)}{u-1}\\
\nonumber  &+ tuF(t, u)^3 + \; tu\frac{F(t, u) - F(t, 1) - (u-1)\partial_uF(t, 1)}{(u-1)^2}.
\end{align}
Note that for any~$a\in\mathbb{Q}$, the \emph{divided difference
operator}~$\Delta_a:F\mapsto (F(t, u) - F(t, a))/(u-a)$ maps~$\mathbb{Q}[u][[t]]$ to itself.
As a consequence, it follows that Equation~\eqref{eqn:funceqn_3const} admits a unique
solution in~$\mathbb{Q}[u][[t]]$: the first fraction is $\Delta_1F$ while the second one is
$\Delta_1^2F\equiv \Delta_1(\Delta_1F)$.
Rewritten as~$F = 1 + tu(3F - (u-1)\Delta_1F)\Delta_1F + tuF^3 + tu\Delta_1^2F$,
Equation~\eqref{eqn:funceqn_3const} is a~\emph{discrete differential equation} (DDE)
of order~$2$, since the operator~$\Delta_1$ is iterated~$2$ times.\\
Equation~\eqref{eqn:funceqn_3const} has the
property that its unique solution~$F\in\mathbb{Q}[u][[t]]$ has a specialization~$F(t, 1)$
which is algebraic over~$\mathbb{Q}(t)$. More precisely, in~\eqref{eqn:funceqn_3const}
the specialization~$F(t, 1) = 1 + t + 6t^2 + 54t^3 + 594t^4 + \cdots$ is a root in~$z$
 of~$81 t^{2} z^{3}-9 t \left(9 t -2\right) z^{2}+\left(27 t^{2}-66 t +1\right) z -3 t^{2}+47 t -1$.

\medskip {\em A general algebraicity result.}
The algebraicity of the unique solution in~$\mathbb{Q}[u][[t]]$
of~\eqref{eqn:funceqn_3const}
is in fact a consequence of the following strong and elegant result proved in~$2006$ by
Bousquet-Mélou and Jehanne. It ensures algebraicity of solutions
of the most frequent class of DDEs of arbitrary order~$k$ and with one catalytic variable, 
namely the class of DDEs of the fixed-point type.
  \begin{theorem} \emph{(\cite[Thm.~3]{BMJ06})}\label{thm:BMJ06}
Let~$\K$ be a field of characteristic~$0$
and consider two polynomials
$Q \in \K[x, y_1, \ldots, y_k, t, u]$
and
$f \in \K[u]$,
where $k\in\mathbb{N}\setminus \{ 0 \}$.
Let $a\in\K$ and
$\Delta_a :\K[u][[t]]\rightarrow \K[u][[t]]$ be the
divided difference operator $\Delta_a F(t, u) := (F(t, u) - F(t,
a))/(u-a)$.
Let us denote by $\Delta_a^{i}$ the operator obtained by iterating $i$ times $\Delta_a$.
Then, there exists a unique solution $F\in \K[u][[t]]$ of the  equation
\begin{equation}\label{eqn:initfunceqn}
F(t, u) = f(u) + t \, Q(F(t, u), \Delta_a F(t, u), \ldots, \Delta_a^{k}F(t, u), t, u),
\end{equation}
and moreover $F(t,u)$ is algebraic over $\K(t, u)$.
\end{theorem}

\cref{thm:BMJ06} has been further extended in~\cite{NoSe22} to
the case of \emph{systems of DDEs} of the form~\eqref{eqn:initfunceqn} with
$1$ catalytic variable.
In fact, the algebraicity results proved in~\cite{BMJ06, NoSe22} are
particular cases of a much deeper and older result in commutative algebra proved
by Popescu~\cite{Popescu86} in the context of Artin approximation theory with
nested conditions. The strength of the approaches presented in~\cite{BMJ06,
  NoSe22} lies in the effectivity of their algebraicity proofs. Despite a recent
improvement in the linear case~\cite{CJPR19}, Popescu's result is still not
known to admit an effective proof.

\medskip    \indent{\em Setting and main goal.}
    In the remainder of this article, we focus only on DDEs of the form~\eqref{eqn:initfunceqn}.
    Note that one can consider the associated \emph{polynomial functional equation} obtained
by multiplying \eqref{eqn:initfunceqn} by the smallest power of $(u-a)$ such that
the product becomes polynomial in~$u$. This new equation is denoted by
\begin{equation}\label{eqn:initpoleqn}
P(F(t, u), F(t, a), \ldots, \partial_u^{k-1}F(t, a), t, u) = 0,
\end{equation}
for some nonzero polynomial~$P\in\mathbb{K}[x, z_0, \ldots, z_{k-1}, t, u]$.

Starting from~\eqref{eqn:initpoleqn}, our main goal is to compute a
nonzero~$R\in\mathbb{K}[t, z_0]$ such that~$R(t, F(t, a)) = 0$.
Remark that setting
$u=a$ in~\eqref{eqn:initpoleqn} yields a tautology, and that
differentiating~\eqref{eqn:initpoleqn} 
with respect to $t$ yields a sum of~$k+2$ terms and introduces~$k+1$ series
from which  
nothing can be deduced, a priori.
In this article, we will focus on designing systematic algorithms for solving equations such as~\eqref{eqn:initpoleqn}.

\medskip {\em Previous work.}
We use the notation~$\overline{\mathbb{K}}[[t^{\frac{1}{\star}}]]\equiv
\cup_{d\geq 1}\overline{\mathbb{K}}[[t^{\frac{1}{d}}]]$.
For \emph{linear} DDEs, the \emph{kernel method}~ (which already appears
in an exercise of Knuth's book~\cite[Ex.~2.2.1-4]{Knuth68}
and was systematized in~\cite{BaFl02})
consists in finding the roots
$u = U(t)\in\overline{\mathbb{K}}[[t^{\frac{1}{\star}}]]$ of the coefficient in~$x$ of~$P$.
By taking the resultant with respect to~$u$ of this coefficient and of~$P$, one obtains
a polynomial relation relating $F(t, a), \ldots, \partial_u^{k-1}F(t, a)$.
Since the work~\cite{BoPe00} of Bousquet-M\'elou and Petkov{\v{s}}ek, 
the linear case can be considered as fully understood.\\
An extension
of this method to the setting where~$\deg_x(P) = 2$ is also classical and is called
the~\emph{quadratic method}. It first appears in Brown's work~\cite{Brown65b}
for the case~$k=1$. It is based on a different elegant argument which
produces a polynomial relation between these specialized series.
This method was extended thirty years later on a particular
family of examples by Bender and Canfield~\cite{BeCa94}.\\
\indent
Both the kernel method and the quadratic method were generalized by the approach proposed
by Bousquet-Mélou and Jehanne in~\cite{BMJ06}.
Their method consists of starting from~\eqref{eqn:initpoleqn} and of creating more
polynomial equations having a nontrivial solution with~$F(t, a)$ as
its~$\{z_0\}$-coordinate. When this
method produces as many polynomials as variables and if the induced system generates
a $0$-dimensional ideal, a polynomial elimination strategy performed on the
polynomial system allows one to compute a nonzero element~$R\in\mathbb{K}[t, z_0]$
annihilating~$F(t, a)$. In the case where the system does not have the above properties,
a deformation of~\eqref{eqn:initfunceqn} via the introduction of a parameter
allows one to compute such an~$R$. \\
\indent This unified method contains however some intrinsic limitations due to the
number of variables introduced when creating more polynomial equations, and to the lack
of geometric interpretation of the problem. It was already mentioned
in~\cite[\S12]{BMJ06}
that the method was ``lacking an efficient elimination theory for
polynomial systems
which (\ldots) are highly symmetric''.\\
\indent A first step in the algorithmic study of DDEs has been
initiated in~\cite{BoChNoSa22} for DDEs 
of order~$1$.

\medskip {\em Main results.}
In constrast with~\cite{BoChNoSa22}, the purpose of the present article is
  to entirely treat the challenging case of DDEs of order~$k\geq 1$.
In~\cref{sec:generic_modeling}, we recall the polynomial system reduction from~\cite{BMJ06}
and provide a geometric interpretation of it.
In~\cref{sec:direct}, we prove under genericity assumptions on the input DDE~\eqref{eqn:initpoleqn}
that the algebraicity degree of~$F(t,a)$ is bounded
by~$\delta^{3k}/k!$~(Prop.~\ref{prop:bounds_under_H1}).
Here and in all that follows, $\delta$ denotes an upper bound on the total degree of~$P$  in~\eqref{eqn:initpoleqn}.
We deduce from this bound that a nonzero
annihilating polynomial of~$F(t, a)$ can be computed
in~$\tilde{O}(\delta^{6k}(k^2\delta^{k+3}+ \delta^{1.89k}/k!))$ ops.
in~$\mathbb{K}$ (Prop. \ref{prop:complexity_under_H1}).
Here, and in the whole paper, the soft-O notation~$\tilde{O}(\cdot)$ 
hides polylogarithmic factors in the argument.
 In~\cref{sec:guessandprove},
we use the upper-bound $\delta^{3k}/k!$ to generalize Prop.~2.11 in~\cite{BoChNoSa22}
and deduce a complexity estimate
in~$\tilde{O}({k\cdot\delta^{10.12\cdot k}}/{(k-1)!^2})$ ops. in~$\mathbb{K}$
(Prop.~\ref{prop:hybridgp}).
In~\cref{sec:elimination}, we introduce a new
algorithm based on algebraic elimination and 
specialization properties of Gr\"obner bases. 
In~\cref{sec:stickelberger}, we
design one more new algorithm based on a geometric interpretation of the problem.
In~\cref{sec:experiments}, we provide 
experimental results based on efficient implementations
of~\cref{sec:direct,sec:guessandprove,sec:elimination,sec:stickelberger}
for several DDEs coming from combinatorics.
The practical gains compared to the state-of-the-art go from a few minutes to several days
of computation time,
 and we solve 
 the DDE related to~$5$-constellations ($k=4$)
 using a combination of~Sections~\ref{sec:guessandprove} and~\ref{sec:elimination}.

\medskip{\em Notation.} 
We denote by~$\mathbb{K}$ an effective field of characteristic~$0$.
We write $\overline{\K}$ for an algebraic
closure of $\mathbb{K}$, and $\K[t], \K(t)$ and $\K[[t]]$ 
for, respectively, the rings of polynomials, rational functions and 
formal power series 
in~$t$ with coefficients in~$\K$.
We also use the notation
${\K}[[t^{1/\star}]]:=\bigcup_{d\geq1} \K[[t^{1/d}]]$ for the ring of~``fractional power
series'', that is series of the form 
$\sum_{n\geq 0} u_n t^{n/d}$ for some
integer $d\geq 1$.
We use the convention $\partial_x f$ and alike for the partial derivative of a function $f$
with respect to~$x$.  For~$p$ a polynomial in~$n$ variables over~$\K$,
we denote by $\disc_x(p)$
its discriminant with respect to a variable~$x$, 
by~$\deg(p)$ its total degree
and by~$\deg_x(p)$ its degree w.r.t. the variable~$x$.
For an ideal $I \subset \K[x_1,\ldots, x_n]$, we denote by
$V(I)= V_{\overline{\K}}(I)$
the affine variety, or the zero set,  defined by $I$ over $\overline{\K}$.

\medskip {\em Polynomial elimination basics.} We will repeatedly make use of
the following fundamental results in polynomial elimination theory. For proofs
and further context, we refer to \cite[Chap.~3, Thm.~2, p.~122]{CoLiOS07} for
Fact~\ref{fact:elim}\eqref{fact:elimination}, \cite[Chap. 3, \S 5, Theorem 3, p.
159]{CoLiOS07} for Fact~\ref{fact:elim}\eqref{fact:extension},
\cite[Theorem~1.1]{Cox20} for Fact~\ref{fact:elim}\eqref{fact:eigenvalue} and
\cite[Theorem~1.2]{Cox20} for Fact~\ref{fact:elim}\eqref{fact:Stickelberger}.

\begin{fact}\label{fact:elim}
Let $\K$ be a field and let $I \subset \K[x_1,\ldots, x_n]$ be an ideal.

\begin{enumerate}
	 \item (Elimination theorem) \label{fact:elimination}
Let $G$ be a Gr\"obner basis of $I$ with respect to the \emph{lexicographic} order with $x_1 \succ \cdots \succ x_n$. Then, for any $0 \leq \ell \leq n$, the set
$G_\ell = G \cap \K[x_{\ell+1},\ldots,x_n]$ is a Gr\"obner basis of the $\ell$-th elimination ideal $I_\ell = I \cap \K[x_{\ell+1},\ldots,x_n]$.

\item (Extension theorem) \label{fact:extension}
Assume that $\K$ is algebraically closed.
Let $I=(f_1,\ldots,f_s)$, with  $f_i = c_i(x_2,\ldots,x_n)x_1^{N_i} + $ (lower degree terms in $x_1$). If $(a_2,\ldots,a_n) \in V(I_1)\setminus V(c_1,\ldots,c_s)$,
then there exists $a_1\in\K$ such that $(a_1,\ldots,a_n)\in V(I)$.
	
\item (Eigenvalue Theorem)\label{fact:eigenvalue}
Assume $I$ is zero-dimensional, let $\mathbb{A} = \K[x_1,\ldots, x_n]/I$
and $m_f : \mathbb{A} \rightarrow \mathbb{A}$ the multiplication-by-$f$ endomorphism of $\mathbb{A}$.
Then, 
the eigenvalues of $m_f$ are the values of $f$ at the finitely many points of $V(I)$.

\item (Stickelberger's theorem)\label{fact:Stickelberger}
If $I$ is radical and under the assumptions of \eqref{fact:eigenvalue},
the characteristic polynomial $\det(x I - m_f)$ of $m_f$ is equal to
$\prod_{a\in V(I)} (x-f(a))$.
\end{enumerate}
\end{fact}	

\smallskip {\em Complexity basics.}  
The algorithmic costs are estimated by counting elementary arithmetic operations
$(+, -, \times, \div)$ in the base field~$\mathbb{K}$ at unit cost. The
notation~$\theta$ refers to any feasible exponent for matrix multiplication over~$\K$.
The best current upper-bound is
$\theta<2.37188$~\cite{DuWuZh22}.
All classical operations on univariate polynomials of degree~$d$ in~$\mathbb{K}[x]$
(multiplication, multipoint evaluation and interpolation, extended gcd, resultant, squarefree part, etc) can be
performed in softly linear time~$\tilde{O}(d)$.
We refer to the  book by von zur Gathen
and Gerhard~\cite{GaGe13} for these facts and related
questions.

\section{From combinatorics to polynomials}\label{sec:generic_modeling}

\subsection{Solving DDEs via polynomial systems} 

The method of Bousquet-Mélou and Jehanne \cite[Section~2]{BMJ06} is based on the idea
of creating,
starting from the input equation~\eqref{eqn:initpoleqn}, 
new polynomial equations inducing a polynomial system
that admits a solution which has $F(t, a)$ as its $z_0$-coordinate.

The corresponding procedure is the following. 
One takes the derivative of \eqref{eqn:initpoleqn} with respect to the catalytic variable~$u$ and finds
\begin{align}\label{eqn:derivative_poleqn}
  \partial_uF(t, u)\;\cdot&\;
  \partial_xP(F(t, u), F(t, a), \ldots, \partial_u^{k-1}F(t, a), t, u)\\
\nonumber & + \partial_uP(F(t, u), F(t, a), \ldots, \partial_u^{k-1}F(t, a), t, u) = 0.
\end{align}

  Now,
  for any solution $u = U(t)$ in $\overline{\mathbb{K}}[[t^{\frac{1}{\star}}]]\setminus{\{a\}}$
  of the equation
\begin{equation}\label{eqn:crucial_eqn_BMJ}
\partial_xP(F(t, u), F(t, a), \ldots, \partial_u^{k-1}F(t, a), t, u) = 0,
\end{equation}
one obtains by plugging $u = U(t)$ in~\eqref{eqn:derivative_poleqn} that
\begin{align*}
  (x, z_0, \ldots, z_{k-1}, u) &=
  (F(t, U(t)), F(t, a), \ldots, \partial_u^{k-1}F(t, a), U(t))\\
  &\in\overline{\mathbb{K}}[[t^{\frac{1}{\star}}]]\times \mathbb{K}[[t]]\times \cdots
  \times \mathbb{K}[[t]]\times \overline{\mathbb{K}}[[t^{\frac{1}{\star}}]]
\end{align*}
is a solution of the following polynomial system
\begin{equation}\label{eqn:initial_system}
 \begin{cases}\;\;\; P(x, z_0, \ldots, z_{k-1}, t, u) = 0,\\
  \partial_xP(x, z_0, \ldots, z_{k-1}, t, u) = 0,\;\;\; u\neq a.\\
  \partial_uP(x, z_0, \ldots, z_{k-1}, t, u) = 0,\\
  \end{cases}
  \end{equation}

Showing the existence of solutions $u = U(t)$ in
$\overline{\mathbb{K}}[[t^{\frac{1}{\star}}]]\setminus{\{a\}}$ to~\eqref{eqn:crucial_eqn_BMJ}
  is generally done in combinatorics by applying
  \cite[Theorem~2]{BMJ06} and by computing the first terms of the solutions to show that
  they are not constantly equal to $a$.
  In order to avoid those checks, we make the following assumption.
\begin{align*}
\label{hyp1} \tag{\bf H1}
\underline{\textsf{Hypothesis 1:}} \quad
  &\deg_u(\partial_xP(x, z_0, \ldots, z_{k-1}, 0, u))\geq k,\\
  &\partial_{y_k}Q\left( f(a), f'(a), \ldots, \frac{f^{(k)}(a)}{k!}, 0, a \right)\neq 0.
\end{align*}

Under~\eqref{hyp1}, we show that there exist~$k$ distinct solutions~$u$
to the constraints given by~\eqref{eqn:initial_system}.
Note that~\eqref{hyp1} holds for a generic choice of $f$ and $Q$
in~\eqref{eqn:initfunceqn}.

\begin{prop}\label{prop:root_existence}

  Under assumption~\eqref{hyp1}, there exist $k$ distinct solutions
  $U_1(t), \ldots, U_k(t)$ in~$u$
  to~\eqref{eqn:crucial_eqn_BMJ}.
  Moreover, all of them are distinct from~$a$ and lie in $\mathbb{L}[[t^\frac{1}{k}]]$,
  where $\mathbb{L}/\mathbb{K}$ is a field extension of degree upper bounded by $k$.

\end{prop}
\begin{proof}
  
    By expressing equation~\eqref{eqn:crucial_eqn_BMJ} in terms of the derivatives of~$Q$
    and by using the first part of~\eqref{hyp1},
    searching solutions $u \neq a$ of~\eqref{eqn:crucial_eqn_BMJ}
    is seen to be equivalent to looking for solutions $u\neq a$ of
 \begin{align*}
   (u&-a)^k
   =\,\, t\cdot(u-a)^k
   \cdot\partial_x Q(F(t, u), \Delta_a F(t, u), \ldots, \Delta_a^k F(t, u), t, u)\\
   \nonumber  &+
   t\cdot \sum\limits_{i = 1}^{k}{(u-a)^{k-i}\partial_{y_i}Q(F(t, u), \Delta_a F(t, u),
      \ldots, \Delta_a^k F(t, u), t, u)}.
 \end{align*}
 
   By specializing this equation
   at $t=0$, it follows that the
   constant term of any solution in $u$ of~\cref{eqn:crucial_eqn_BMJ} is equal to $a$.
   Then,
   taking the $k$-{th} root of the previous equation yields an identity of the form
   $u = a + t^{\frac{1}{k}}(\alpha + X)^{\frac{1}{k}}$, where
   $\alpha := \partial_{y_k}Q(f(a), f'(a), \ldots, \frac{1}{k!}f^{(k)}(a), 0, a)
   \in\mathbb{K}\setminus\{0\}$ and $X\in\mathbb{K}[[t]][[u]]$ satisfies $X(0, a) = 0$.
   By applying Newton's generalized binomial theorem,
   one expands $(\alpha + X)^{\frac{1}{k}}$
   in~$\mathbb{L}[[t]][[u]]$, with~$\mathbb{L} = \mathbb{K}(\alpha^{\frac{1}{k}})$.
   By a fixed-point argument applied to $u = a + t^{\frac{1}{k}}(\alpha + X)^{\frac{1}{k}}$,
   the $k$-{th} roots of~$\alpha$ induce
   the existence of~$k$ distinct solutions
   $U_1(t), \ldots, U_k(t)\in\mathbb{L}[[t^{\frac{1}{k}}]]\setminus\{a\}$
   in~$u$
   to~\eqref{eqn:crucial_eqn_BMJ}, all of them lying in~$\mathbb{L}[[t^\frac{1}{k}]]$.
\end{proof}

The main idea of~\cite{BMJ06} is that the existence of the distinct solutions
$U_1(t), \ldots, U_k(t)$ induce distinct pairs
  $(F(t, U_i(t)), U_i(t))\in\overline{\mathbb{K}}[[t^{\frac{1}{k}}]]^2$
  for every $1\leq i \leq k$.
  Hence the point
  \begin{align}\label{eqn:point_I}
    (&x_1, u_1, \ldots, x_k, u_k, z_0, \ldots, z_{k-1})
    =\\
    \nonumber &(F(t, U_1(t)), U_1(t), \ldots, F(t, U_k(t)), U_k(t), F(t, a),
    \ldots, \partial_u^{k-1}F(t, a))
  \end{align}
  is a solution of the \emph{duplicated} polynomial system
  \begin{equation}\label{eqn:duplicated_system}
    \forall \;1\leq i \leq k\;, 
    \begin{cases} \;\;\;P(x_i, z_0, \ldots, z_{k-1}, t, u_i) = 0,\\
  \partial_xP(x_i, z_0, \ldots, z_{k-1}, t, u_i) = 0,\\
  \partial_uP(x_i, z_0, \ldots, z_{k-1}, t, u_i) = 0,\\
  \end{cases}
  \end{equation}
  defined by $3k$ equations in $3k$ variables over $\mathbb{K}(t)$.
  Now, to avoid the irrelevant solutions of~\eqref{eqn:duplicated_system},
  we restrict our attention to the solutions of~\eqref{eqn:duplicated_system}
  that are not solutions of
  $\prod_{i\neq j}({u_i - u_j})\cdot\prod_{i}({u_i - a}) = 0$; we define 
  $\textsf{diag}\in\mathbb{K}[u_1, \ldots, u_k]$ as the left-hand side of this equation.
  \begin{notation}
  We write $\underline{x}$ (resp. $\underline{u}$ and $\underline{z}$) for the
  variables  $x_1, \ldots, x_k$ (resp. $u_1, \ldots, u_k$ and $z_0, \ldots, z_{k-1}$) 
  and  $\mathcal{I}$ for the ideal
  of~$\mathbb{K}(t)[\underline{x}, \underline{u}, \underline{z}]$ generated
  by~\eqref{eqn:duplicated_system}.
  \end{notation}
  With the extra condition that
  $\textsf{diag}\neq 0$, those~$3k$ equations and variables generically define
  a $0$-dimensional ideal over~$\mathbb{K}(t)$ and hence induce 
  a finite set of solutions. For a later effective use of this finiteness, we introduce
  the following regularity assumption.
\begin{align*}
\label{hyp2} \tag{\bf H2}
\underline{\textsf{Hypothesis 2:}} \quad
&\text{The ideal }\mathcal{I}_\infty := \mathcal{I}:\textsf{diag}^\infty\text{ in }
\mathbb{K}(t)[\underline{x}, \underline{u}, \underline{z}]\\
&\text{ is radical }\text{and has dimension}~0 \text{ over } \mathbb{K}(t).
\end{align*}

      \indent
      We now show that eliminating all variables in~$\mathcal{I}_\infty$ except~$t$
      and~$z_0$ yields a nonzero annihilating polynomial of~$F(t, a)$.
  
  \begin{prop}\label{prop:trivial_elimination_property}
    Under~\eqref{hyp1} and~\eqref{hyp2}, if
    $R\in\mathcal{I}_\infty\cap\mathbb{K}[t, z_0]\setminus\{0\}$
    then~$R(t, F(t, a)) = 0$.
  \end{prop}
  \begin{proof}
    Under~\eqref{hyp1}, we apply Proposition~\ref{prop:root_existence} to justify that
    the point given by~\eqref{eqn:point_I} lies in~$V(\mathcal{I})$. Now by definition of
    a saturated ideal, there exists $m\in\mathbb{N}$ such that
    $\textsf{diag}^m\cdot R\in\mathcal{I}$ writes as an algebraic expression
    in the polynomials involved in~\eqref{eqn:duplicated_system}.
    Specializing this expression to the point given by~\eqref{eqn:point_I} and using
    Proposition~\ref{prop:root_existence}, the point given by~\eqref{eqn:point_I}
    does not annihilate $\textsf{diag}$. Hence it annihilates~$R$.
    Finally,  the dimension property in~\eqref{hyp2} implies
    that $\mathcal{I}_\infty\cap\mathbb{K}[t, z_0]$
    is not reduced to~$\{0\}$.
  \end{proof}

  \subsection{Geometric interpretation} 
  We now introduce a geometric interpretation of the fact
  that~\eqref{eqn:point_I} is a solution of~\eqref{eqn:duplicated_system}.
   Recall that a subset of ${\acfield{t}}^{k}$
is said to be \emph{constructible} if
it is  a finite union of Zariski open subsets of a Zariski closed subset of~${\acfield{t}}^{k}$.
Typically,
  a set defined by polynomial equations and inequations is constructible.
  Hence denoting by~$\mathcal{X}\subset\acfield{t}^{k+2}$ the set defined by the
  constraints~\eqref{eqn:initial_system}, we have that~$\mathcal{X}$ is a constructible set.
  We now define new geometric objects and assumptions for any constructible set
  $\scrW\subset {\acfield{t}}^{k+2}$ associated with polynomial constraints
  in~$\K(t)[x, u, \underline{z}]$, and deduce simple properties when~$\scrW = \mathcal{X}$.
  Define the
canonical projection $\pi : (x, u, \underline{z})\in\acfield{t}^{k+2}
\mapsto(\underline{z})\in\acfield{t}^k$ onto
the $\underline{z}$-coordinate space. In the whole paper, we assume that:
\begin{align*}
  \label{hyp:finite} \tag{\bf F}
  \text{The restriction of $\pi$ to $\scrW$ has finite fibers}. 
\end{align*}
For $\balpha\in \acfield{t}^{k}$, we denote by $\scrW_{\balpha}\subset
\acfield{t}^{2}\times \acfield{t}^{k}$ the fiber $\pi^{-1}(\balpha)\cap \scrW$,
by $\card{u}{\scrW}{\balpha}$ the number of $u$-coordinates of the points in
$\scrW_\balpha$. We set $\cfiber{k}{u}{\scrW} := \{{\balpha}\in{\acfield{t}}^{k} \mid \card{u}{\scrW}{\balpha} \geq k\}$. 
      \begin{lemma}\label{lemma:constructibility_Fk}
		If $\scrW\subset {\acfield{t}}^{k+2}$ is constructible, then
		$\cfiber{k}{u}{\scrW}$  is also constructible.
                Moreover, under~\eqref{hyp1},
		$\cfiber{k}{u}{\mathcal{X}}$ is not empty.
        \end{lemma}
      \begin{proof}
        By fixing the variables~$\underline{z}$ and duplicating~$k$ times the
        variables~$x$ and $u$, it is possible to define relevant equations
        ensuring at least~$k$ solutions with distinct~$u$ coordinates. Now,
        eliminating all variables but~$\underline{z}$ and using~\cite[Thm.~7,
        \S7, Ch.~4, p.~226]{CoLiOS07} is enough to deduce that the
        projection of the solution set of these duplicated constraints onto the
        $\underline{z}$-coordinate space is a constructible~set.
           Under~\eqref{hyp1},~\cref{prop:root_existence} implies
        that the system~\eqref{eqn:initial_system} admits (at least)
        $k$ solutions in~$\overline{\mathbb{K}}[[t^\frac{1}{\star}]]$ with
        same~$\underline{z}$-coordinates, and distinct $u$-coordinates. This proves
		that
         $\cfiber{k}{u}{\mathcal{X}}$ is not empty.\looseness=-1
      \end{proof}

      The aim of the new algorithm that we will introduce in
      Section~\ref{sec:elimination} is to compute a disjunction of 
      conjunctions of polynomial equations and inequations
      in~$\mathbb{K}(t)[\underline{z}]$ whose solution set in~${\acfield{t}}^k$
      is $\cfiber{k}{u}{\mathcal{X}}$.
		
        Denoting~$z_0, z_2, \ldots, z_{k-1}$ by~$\check{z}_1$,
        we now consider the projection map~$\pi_{\check{z}_1}:(x, u, \underline{z})
        \in\acfield{t}^{k+2}
        \mapsto (\check{z}_1)\in\acfield{t}^{k-1}$. We assume in the rest
        of this paper that the following assumption holds:
        \begin{align*}
  \label{hyp:stick} \tag{$\mathbf{\check{F}}$}
  \text{The restriction of $\pi_{\check{z}_1}$ to $\scrW$ has finite fibers}. 
\end{align*}

        Also, we introduce a set~$\stickW$
		that will provide a second geometric interpretation of our problem, 
		and will yield a second algorithm, given in~\cref{sec:stickelberger}. 
		We thus define
        \begin{align*}
          \stickW := \{\balpha =
          \;&(\alpha_0, \ldots, \alpha_{k-1})\in\overline{\mathbb{K}(t)}^k|
          \;\; \balpha\in\pi(\scrW) \;\;\wedge\;\; \\
          &\#\scrW \cap  \pi_{\check{z}_1}^{-1}((\alpha_0, \alpha_2, \ldots, \alpha_{k-1})) \geq k\}.
        \end{align*}
        \begin{lemma}\label{lemma:SW_constructible}
          The set $\stickW$ is 
          constructible. Moreover, under assumption~\eqref{hyp1}
		  the set $\mathcal{S}_k(\check{z}_1, \mathcal{X})$
          is not empty.
        \end{lemma}
        \begin{proof}
          Considering polynomial constraints defining the set~$\scrW$,
          the cardinality condition is modeled by: fixing the variables $\check{z}_1$,
          duplicating the variables~$x, u, z_1$ and defining a conjunction
          of polynomial constraints
          ensuring the solutions of such a system to be distinct w.r.t. the
          duplicated coordinates.
          By~\cite[Thm.~7, \S7, Ch.~4, p.~226]{CoLiOS07}, we
          deduce that~$\stickW$ is constructible.
          Under~\eqref{hyp1}, the set $\mathcal{S}_k(\check{z}_1, \mathcal{X})$
          contains~$(F(t, a), \ldots, \partial_u^{k-1}F(t, a))\in\acfield{t}^k$.
          \looseness=-1
        \end{proof}
        
        In Sec.~\ref{sec:stickelberger}, we will introduce a new algorithm that computes
        a finite set of polynomial constraints
        in~$\mathbb{K}(t)[\underline{z}]$
        characterizing~$\mathcal{S}_k(\check{z}_1, \mathcal{X})$.
        Note that~$\cfiber{k}{u}{\mathcal{X}}$
        and~$\mathcal{S}_k(\check{z}_1, \mathcal{X})$ are related as follows:

        \begin{lemma}
          The following inclusion holds~$\cfiber{k}{u}{\mathcal{X}}\subset
          \mathcal{S}_k(\check{z}_1, \mathcal{X})$.
        \end{lemma}
        \begin{proof}
          Let us choose~$\balpha = (\alpha_0, \ldots, \alpha_{k-1})\in
          \cfiber{k}{u}{\mathcal{X}}$.
          By definition of~$\cfiber{k}{u}{\mathcal{X}}$ we have
          that~$\balpha\in\pi(\mathcal{X})$, hence
$(\alpha_0, \alpha_2, \ldots, \alpha_{k-1})\in\pi_{\check{z}_1}(\mathcal{X})$. Now,
          any of the~$k$ points in~$\mathcal{X}\cap\pi^{-1}(\balpha)$ also belongs
to~$\mathcal{X}\cap\pi_{\check{z}_1}^{-1}((\alpha_0, \alpha_2, \ldots, \alpha_{k-1}))$.
          Hence~$\balpha\in\mathcal{S}_k(\check{z}_1, \mathcal{X})$.
          \end{proof}

 \section{Direct approach: degree bounds and complexity}\label{sec:direct}

  %
 In this section, we focus on the complexity of computing a
 nonzero element $R$ of $\mathcal{I}_\infty\cap\mathbb{K}[t, z_0]$, by using the
 work of Bousquet-Mélou and Jehanne~\cite{BMJ06}. The
 following analysis is a generalization of the one
 given in~\cite[Proposition~2.8]{BoChNoSa22}. It takes benefit of the group
 action of the symmetric group~$\mathfrak{S}_k$ on the
 zero set~$V(\mathcal{I}_\infty)$, which can be exploited for DDEs of
 order~$k>1$.
  \begin{prop}\label{prop:bounds_under_H1}
  Let $P$ be as in \eqref{eqn:initpoleqn} of total degree $\delta$.
  Assume that~\eqref{hyp1} and \eqref{hyp2} hold.
  Then $\deg(\mathcal{I}_\infty)$ and $\#V(\mathcal{I_\infty})$ are bounded
  by~$\delta^k\cdot(\delta - 1)^{2k}$, and there exists
  $R\in\K[t, z_0] \setminus \{ 0 \}$
  satisfying $R(t, F(t, a))~=~0$, with $\deg_t(R)$ and $\deg_{z_0}(R)$ at most
  $\mathfrak{b}$, where
   $ \mathfrak{b}~:=~{\delta^k(\delta-1)^{2k}}/{k!}$.
\end{prop}
\begin{proof}
  First, we
  identify a nonzero polynomial in ${\mathcal{I}}_\infty\cap\K[t, z_0]$.
  Note that~\eqref{hyp1} and~\eqref{hyp2} allow us to
  apply~\cref{prop:trivial_elimination_property} which implies that such a polynomial
  annihilates~$F(t, a)$.
  Now,~\eqref{hyp2} implies that the quotient ring $\K(t)[\underline{x},
  \underline{u}, \underline{z}]/{\mathcal{I}}_\infty$ is a finite dimensional
  $\K(t)$-vector space. Hence by
  Fact~\ref{fact:elim}\eqref{fact:eigenvalue},
  the endomorphism $m_{z_0}:f\mapsto z_0\cdot f$ admits a
  characteristic polynomial $\xi_{z_0}\in\K(t)[z_0]$ whose roots are exactly
  the $z_0$-coordinates of all points in $V({\mathcal{I}}_\infty)$ (in finite
  number by assumption~\eqref{hyp2}). Hence, multiplying $\xi_{z_0}$ by the lcm of
  the denominators of
  its coefficients and denoting by~$R$ the squarefree part of the resulting polynomial,
  the radicality of $\mathcal{I}_\infty$, together with Hilbert's Nullstellensatz,
  implies that~$R\in\mathcal{I}_\infty\cap
  \mathbb{K}[t, z_0]\setminus\{0\}$. Hence~$R$ satisfies~$R(t, F(t, a)) = 0$.\\
  \indent We now prove that the degrees of $R$ in $t$ and $z_0$ are both
  bounded by $\mathfrak{b}$. We apply the exact same proof as the one done
  for proving~\cite[Proposition~2.8]{BoChNoSa22} but with
  $\mathcal{I}$ replaced by $\mathcal{I}_\infty$ and with
  $\deg(\partial_xP)$ and~$\deg(\partial_uP)$ both bounded by~$\delta-1$. This implies
  that~$\deg(\mathcal{I}_\infty)$ and $\# V(\mathcal{I}_\infty)$ are bounded
  by $\delta^k(\delta-1)^{2k}$.
  As the partial degrees of~$R$ are bounded
  by~$\deg(\mathcal{I}_\infty)$, it follows that~$\deg_{z_0}(R)\leq \delta^k(\delta-1)^{2k}$.\\
  \indent It remains to justify the nontrivial division by $k!$ (which did not
  appear in~\cite[Proposition~2.8]{BoChNoSa22}).
  We exploit the following group action
    of~$\mathfrak{S}_k$ over~$V(\mathcal{I}_\infty)$.
    Denote by~$\pi_{\underline{z}}:\acfield{t}^{3k}\rightarrow
    \acfield{t}^k$ the map such that~$\pi_{\underline{z}}(V(\mathcal{I}_\infty))$
    is the 
    projection of~$V(\mathcal{I}_\infty)$ onto
    the~$\underline{z}$-coordinate space. 
  Let~$\balpha\in \pi_{\underline{z}}(V({\mathcal{I}_\infty}))$,
    and consider any $k$-tuple
    $(\xi_i, \upsilon_i, \balpha)$ (for $1\leq i \leq k$) in $\mathcal{X}\cap
    \pi^{-1}(\balpha)$. Then for all $1\leq i
    \leq k$, and all~$\sigma\in\mathfrak{S}_k$,
    the concatenation of all
    $(\xi_{\sigma(i)}, \upsilon_{\sigma(i)})$ remains a solution to the
    system defining $\mathcal{I}_\infty$, where the
    $\underline{z}$-coordinates are specialized to the coordinates
    of $\balpha$. Since
    $\textsf{diag}\neq 0$, then $\upsilon_i\neq \upsilon_j$ for $i\neq j$. Hence
    the above orbit has cardinality $k!$.
    Since all roots in
    $\overline{\K(t)}$ of~$R$, seen as a polynomial in $\K(t)[z_0]$, correspond
    to one coordinate of $\balpha\in \pi_{\underline{z}}(V({\mathcal{I}_\infty}))$,
    we deduce that $\deg_{z_0}(R)$ is bounded by the cardinality of
    $V(\mathcal{I}_\infty)$ divided by $k!$, the combinatorial complexity of
    $\mathfrak{S}_k$. Bounding $\deg_{t}(R)$ is done the same way, by
    inverting the roles of $z_0$ and $t$.\looseness=-1
  \end{proof}

\begin{prop}\label{prop:complexity_under_H1}
  Let $P$ be as in \eqref{eqn:initpoleqn} of total degree $\delta$. We suppose
  that~~\eqref{hyp1} and~~\eqref{hyp2} hold. Then there exists an algorithm
  which takes as input a straight-line program of length $L$
  evaluating~$P$ and~$\textsf{diag}$, and returns a nonzero polynomial $R\in\K[t,
    z_0]$
  such that~$R(t, F(t,
  a))=0$,  using
  \begin{align*}
    \tilde{O}((kL+1)\delta^{2k}(\delta-1)^{4k} &+ \delta^{2.63k}(\delta-1)^{5.26k}/k!)\\
    &\subset
    \tilde{O}(\delta^{6k}(k^2\delta^{k+3} + \delta^{1.89k}/k!))
  \end{align*}
  arithmetic operations in~$\K$.
\end{prop}
\begin{remark}
Note that the above complexity is polynomial in~$\mathfrak{b}$.
  \end{remark}
\begin{proof}
  We generalize the proof of~\cite[Proposition~2.9]{BoChNoSa22} to our
  situation. By~\cref{prop:bounds_under_H1},
  $\deg(\mathcal{I}_\infty)$ and $\#V(\mathcal{I_\infty})$ are bounded
  by~$\delta^k\cdot(\delta - 1)^{2k}$. Using the algorithm
  underlying~\cite[Theorem~2]{Schost03}, we thus compute
  a parametric geometric resolution \cite{Schost03} of the zero
  set~$V(\mathcal{I}_\infty)$
  in~$\tilde{O}((kL+1)\delta^{2k}(\delta-1)^{4k})$ ops. in~$\mathbb{K}$.
  This algorithm computes two polynomials~$V(t, \lambda), W(t, \lambda)\in
  \mathbb{K}(t)[\lambda]$ giving the
  parametrization~$z_0=V(t, \lambda)/\partial_\lambda W(t, \lambda)$ whenever
  $W(t, \lambda)=0$.
  Now, we define the map~$m_{z_0}:f\mapsto f\cdot z_0$
  in~$\mathbb{K}(t)[\underline{x}, \underline{u}, \underline{z}]/\mathcal{I}_\infty$
  and observe that its characteristic polynomial is the resultant  w.r.t.~$\lambda$
   of~$z_0\cdot\partial_\lambda W(t, \lambda) -
  V(t, \lambda)$ and
  $W(t, \lambda)$. We thus compute the squarefree part~$R$ of this
  resultant: \textit{(i)} by performing evaluation--interpolation on~$t$ with, by
  Proposition~\ref{prop:bounds_under_H1},
  ~$O(\mathfrak{b})$
  points; \textit{(ii)} for each evaluation in~$t=\bm{\theta}\in\K$,
  the polynomials
  $z_0\cdot \partial_\lambda W(\bm{\theta}, \lambda) - V(\bm{\theta}, \lambda)$
  and~$W(\bm{\theta}, \lambda)$ are bivariate polynomials, which allows us to
  use~\cite[\S 5]{HyNeSc19} for the bivariate resultant computation.
  This step is in~$\tilde{O}(\delta^{2.63k}(\delta-1)^{5.26k}/k!)$
  ops.~in~$\mathbb{K}$.
  Finally, the inclusion comes from the cost for evaluating~$P$
  and the saturating polynomial~$\textsf{diag}$, which by the Baur-Strassen
  theorem~\cite[Theorem~1]{BaSt83} satisfies~$L\in O(k\delta^{k+3})$.
  \looseness=-1
\end{proof}

Despite the process of duplicating variables as done for obtaining the
system~\eqref{eqn:duplicated_system} is, up to deforming the initial DDE,
very fruitful for creating zero dimensional ideals and showing theoretical
algebraicity results~\cite{BMJ06,NoSe22}, it usually suffers from efficiency issues.
\noindent
Applying~\cite[Prop.~2]{Heintz83} for~$n, m\in\mathbb{N}$ and an algebraic
set~$V\subset\overline{\mathbb{K}}^n$,
we have that~$\deg(V^m) = \deg(V)^m$. The degenerate behavior of the state-of-the-art
when $k$ grows up comes from this exponential growth of the ideal's degree in the number of
duplications (which is $k$ in our case).
Moreover, duplicating variables also introduces $3k+2$ variables, while the initial system
in~\eqref{eqn:initial_system} only deals with $k+3$ variables.
A natural hope is hence to
avoid these duplications by
a careful analysis of the geometry given by the initial
constraints~\eqref{eqn:initial_system}.

\begin{remark}
  This number of duplications and the group action of~$\mathfrak{S}_k$
over~$V(\mathcal{I}_\infty)$
is usually exploited by the state-of-the-art of the
polynomial system solving theory by working
in the invariant ring associated to this group action (see~\cite{FaSv12}). However, 
in our case, this approach would imply to introduce a number of variables which would be
at least equal to~$k$.
In the algorithm we propose in~\cref{sec:elimination,sec:stickelberger}, we focus
on introducing no more extra variables.
\end{remark}

\section{Hybrid guess-and-prove algorithm}\label{sec:guessandprove}

  \indent
  We analyze in what follows the complexity of the \emph{hybrid guess-and-prove} algorithm
introduced in \cite[\S2.2.2]{BoChNoSa22}.
Recall that it blends algebraic elimination with a guess-and-prove approach
inspired by Zeilberger's method \cite{Zeilberger92}, see \cite[\S2.2.1]{BoChNoSa22}.
For functional equations of arbitrary order, the motivation of this
method comes from certain concrete examples for which the involved polynomial systems are
difficult to solve (e.g. \cite[\S3.6]{BMWa20}). 
Let us recall the algorithm:

\smallskip \underline{Hybrid guess-and-prove method}
\begin{enumerate}
  \item[(0)] Compute $F(t,a)$ mod~$t^\sigma$ for some integer~$\sigma\geq 1$;
  \item[(1)] Guess $R \in \K[t, z_0]\!\setminus\!\{ 0 \}$ such that
    $R(t, F(t,a))=0 \bmod t^\sigma$;
  \item[(2)] Check if $R(t, F(t,a)) = 0 \bmod t^{\mathfrak{b} \cdot \deg R + 1}$;
  \item[(3)] If not, then go back to~(0) with $\sigma:=2\sigma$; if yes, then return~$R$.
\end{enumerate}
  The correctness of this method is a consequence of~\cite[\S2.2.2]{BoChNoSa22}
  and of the existence, under suitable hypothesis, of a nonzero
  polynomial $R\in\mathbb{K}[t, z_0]$
  annihilating $F(t, a)$ with partial degrees bounded by~$\mathfrak{b}$.
  It remains to make those hypotheses completely explicit.
  \begin{notation}
    We still assume the existence
  of~$k$ distinct nonconstant solutions $u=U_1(t), \ldots, U_k(t)\neq a$
  to~\eqref{eqn:crucial_eqn_BMJ} and
  we denote by $\mathcal{P}$ the point of $\overline{\mathbb{K}}[[t^{\frac{1}{\star}}]]^{3k}$
  obtained by concatenating the values
  $\{F(t, U_i(t))\}_{1\leq i\leq k}, \{U_i(t)\}_{1\leq i\leq k}$
  and $\{\partial_u^iF(t, a)\}_{0\leq i\leq k-1}$. 
  \end{notation}
\vspace{-0.5cm}
\begin{align*}
\label{hyp4} \tag{\bf H4}
\underline{\textsf{Hypothesis 4:}} \quad
&\text{The Jacobian matrix } \textsf{Jac} \text{ of }
\eqref{eqn:duplicated_system},\\
&\text{ considered in } \underline{x},
\underline{u}, \underline{z}, 
\text{ is invertible at } \mathcal{P}.
\end{align*}
\begin{lemma}\label{lemma:bounds_under_H3}
  Under~\eqref{hyp1} and \eqref{hyp4}, the saturation
  $\mathcal{I}:\operatorname{det}(\textsf{Jac})^\infty$
  is a radical and $0$-dimensional ideal
  of $\mathbb{K}(t)[\underline{x}, \underline{u}, \underline{z}]$. Hence there exists
  a nonzero
  $R\in\mathcal{I}:\operatorname{det}(\textsf{Jac})^\infty\cap\mathbb{K}[t, z_0]$
  annihilating $F(t, a)$ with partial degrees bounded by~$\mathfrak{b}$.
\end{lemma}
\begin{proof}
  The radicality and dimension results are consequences of~\cite[Lemma~2.10]{BoChNoSa22}.
  The rest of the proof is the same as the one used for proving
  Proposition~\ref{prop:bounds_under_H1}, with $\mathcal{I}_\infty$ replaced by
  $\mathcal{I}:\operatorname{det}(\textsf{Jac})^\infty$.
  \end{proof}
This concludes the correctness of the hybrid guess-and-prove method in the case of DDEs of
order~$k>1$. Under~\eqref{hyp1} and~\eqref{hyp4}, we also deduce a complexity estimate
generalizing~\cite[Prop.~2.11]{BoChNoSa22}.

\begin{prop}\label{prop:hybridgp}
 Let $P$ be as in~\eqref{eqn:initpoleqn} and let $\delta$~be its total degree.
 Let us assume that assumptions \eqref{hyp1} and~\eqref{hyp4} hold, and that there exists a
 straight-line program of length~$L$ evaluating~$P$.
Then the hybrid guess-and-prove method terminates on input~$P$ using
  \begin{equation*}
    \tilde{O}(L\cdot k^3\cdot \mathfrak{b}^2 +\,k\cdot \mathfrak{b}^{\theta+1})
    \subset
    \tilde{O}({k\cdot\delta^{10.12\cdot k}}/{(k-1)!^2})
  \end{equation*}
  arithmetic operations in $\K$.
\end{prop}
\begin{proof}
  We analyze the last execution of steps $(0)$--$(3)$, happening with
  $\sigma = O(\mathfrak{b}^2)$ and $\deg(R) = O(\mathfrak{b})$.
  Using Hermite-Padé approximation (e.g.~\cite{GiJeVi03,BoJeSc08}), the guessing at
  step $(1)$ is done in $\tilde{O}(\sigma\cdot\deg_{z_0}(R)^{\theta-1})
  \subset\tilde{O}(\mathfrak{b}^{\theta+1})$ ops. in~$\K$.
  The order of $F(t, a)$ and the truncation order $\mathfrak{b}\cdot\deg(R) + 1$ at
  step $(2)$ are in~$O(\mathfrak{b}^2)$.
  Hence the truncated evaluation of $R$ at $z_0 = F(t, a)$
  is computed in $\tilde{O}(\deg(R)\cdot(\mathfrak{b}\cdot\deg(R)))\subset
  \tilde{O}(\mathfrak{b}^3)\subset \tilde{O}(\mathfrak{b}^{\theta+1})$ ops. in~$\K$.\\
   \indent  Let~$V$ the system given by \eqref{eqn:duplicated_system}. 
  We compute the truncated series $F(t, a)$ using the classical Newton method, by considering
  the iteration
  \begin{equation}\label{eq:newton}
\bF \mapsto \mN(\bF) := \bF - J(\bF)^{-1} \cdot V^{\mathrm{T}} (\bF) ,
\end{equation}
  where
  $J = \textsf{Jac}(V)$ (w.r.t. $\underline{x}, \underline{u}, \underline{z}$)
and $\bF \in \mathbb{L}[[t^\frac{1}{k}]]^{3k}$ denotes an approximation of the solution
$\mathcal{F}$ at which $J$ is invertible by assumptions~\eqref{hyp1} and~\eqref{hyp4},
with $\mathbb{L}/\mathbb{K}$ a finite field extension
of degree at most~$k$ (by Prop.~\ref{prop:root_existence}).
Before going further, recall that the cost for elementary arithmetic operations in
$\mathbb{L}$
can be expressed in terms of elementary arithmetic operations in~$\mathbb{K}$: as the
field extension is
of degree at most~$k$, multiplying and summing two elements of~$\mathbb{L}$ can be
done in
$\tilde{O}(k)$ arithmetic operations in~$\mathbb{K}$.

As the number of correct terms is doubled at each iteration
of~\eqref{eq:newton}, we iterate \eqref{eq:newton} at most
$O(\log(k))$ times to obtain a truncation of 
$\mathcal{F} \bmod t$ (because of the ramification appearing in the Puiseux series~$U_i$).
Hence we perform this precomputation before doubling the integer
power in $t$.

By the Baur-Strassen theorem \cite{BaSt83}, a straight-line program of length $O(L\cdot k)$
evaluating $V$ can be obtained from the one evaluating~$P$. By iterating this argument,
one also finds a straight-line program of length $O(L\cdot k)$
that evaluates the Jacobian matrix.
Consequently, there is also one in $O(L\cdot k)$ for its inverse.
Evaluating~$V$ and $J^{-1}$ at~$\bF$ of some order $N$ requires consequently
$\tilde{O}(L\cdot k\cdot N)$ arithmetic operations in $\mathbb{L}$,
because of the fact that the
cardinality of the support of
$U_i \bmod t^N$ is in $O(k\cdot N)$. This is also the cost of a Newton iteration.
All in all, one obtains a complexity for step $(0)$ in $\tilde{O}(L\cdot k^2\cdot \sigma)$
arithmetic operations in $\mathbb{L}$.
Summing all complexities, the hybrid guess-and-prove complexity is
in~$\tilde{O}(L\cdot k^2\cdot \mathfrak{b}^2 +\, \mathfrak{b}^{\theta+1})$ ops.
in~$\mathbb{L}$.
This gives the global complexity
of~$\tilde{O}(L\cdot k^3\cdot \mathfrak{b}^2 +\, k\cdot \mathfrak{b}^{\theta+1})$ ops.
in~$\mathbb{K}$.
The inclusion is a consequence of the estimate $L\in O(\delta^{k+3})$.
\end{proof}

\section{Approach using elimination theory}\label{sec:elimination}

%
\def\idp{\mathcal{H}}
\def\scrW{\mathscr{W}}

 Let $\scrW\subset \overline{\mathbb{K}(t)}^{k+2}$ be a {constructible} set
  defined by polynomial constraints in $\mathbb{K}(t)[x, u, \underline{z}]$.
We assume that assumption~\eqref{hyp:finite} holds.
For $i\in \mathbb{N}$, we consider the (possibly infinite) set
\[
\cfiber{i}{u}{\scrW} \coloneqq
  \{\bm{\alpha}\in{\acfield{t}}^{k} \mid \card{u}{\scrW}{\balpha} \geq i\}.
\]

Observe that $\cfiber{i + 1}{u}{\scrW}\subseteq \cfiber{i}{u}{\scrW}$ for any
$i\geq 1$.
Adapting easily the proof of~\cref{lemma:constructibility_Fk} yields that
$\cfiber{i}{u}{\scrW}$ is constructible.

In this section, we provide an algorithm that takes as input $i\in \mathbb{N}$
and a polynomial system defining some algebraic set $W\subset
\acfield{t}^{k+2}$ and returns a disjunction of conjunctions of polynomial
  equations and inequations in $\field{t}[\underline{z}]$ whose solution set in
$\acfield{t}^{k}$ is~$\cfiber{i}{u}{W}$.
We then show how to apply this algorithm to compute witnesses of
algebraicity to solutions of DDEs of order~$k$. \looseness=-1

To begin with, we assume that $W$ is given by a polynomial sequence $\bmf$
in $\mathbb{K}[t][x, u, \underline{z}]$ and we denote by $\mathcal{J}$ the ideal
it generates in~$\mathbb{K}(t)[x, u, \underline{z}]$.
To design our algorithm, we leverage advanced results of the theory of Gröbner
bases to characterize $\cfiber{i}{u}{W}$.

Let $\rho_x: (x, u, \underline{z}) \mapsto (u, \underline{z})$ be the canonical
projection which forgets the variable~$x$. We denote by $G$ a Gröbner basis for
{$(\mathcal{J}, \succ)$}, where $\succ$ is a lexicographic monomial ordering
with $x \succ u \succ \uline{z}$. Let $G_x = G\cap \field{t}[u, \uline{z}]$ and
$\ell_x$ be the leading coefficients w.r.t. the variable $x$ of the polynomials
in $G$ which have positive degree w.r.t.~$x$. Finally, we extend the
  definition of~$\cfiber{i}{u}{\cdot}$ to constructible sets defined with
  constraints in~$\K(t)[u, \underline{z}]$ and we let $\scrW = \rho_x(W)$.

\begin{lemma}\label{lemma:transfer}
  The set $\cfiber{i}{u}{W}$ coincides with $\cfiber{i}{u}{\scrW}$. 
\end{lemma}

\begin{proof}
  Let $\balpha\in\pi(W)$. As the $u$-coordinates of $W_\balpha$
  coincide with the $u$-coordinates of the points in $\scrW$ projecting on~$\balpha$, 
  the conclusion follows from the definitions of
  $\cfiber{i}{u}{W}$~and~$\cfiber{i}{u}{\scrW}$.\looseness=-1
\end{proof}

\begin{lemma}\label{lemma:proj_x}
  The set $\scrW$ is defined by the vanishing of all polynomials in $G_x$ and
  the nonvanishing of at least one element in~$\ell_x$.
\end{lemma}

\begin{proof}
  This follows by applying Fact~\ref{fact:elim}\eqref{fact:elimination} and
  Fact~\ref{fact:elim}\eqref{fact:extension}.\looseness=-1
\end{proof}

We use Lemmas~\ref{lemma:transfer} and~\ref{lemma:proj_x} to compute a
polynomial system that encodes $\cfiber{i}{u}{\scrW}$. Remark that by applying
Lemma~\ref{lemma:proj_x}, $\scrW$ is the union of the locally closed sets
defined by the vanishing of all polynomials in $G_x$ and the nonvanishing of at
least one element of~$\ell_x$. Note that the vanishing set of $G_x$ is the
Zariski closure of $\scrW$. Furthermore, we denote by $G_u$ the set $G_x\cap
\field{t}[\uline{z}]$.

Given a set $F$ of polynomials in $\field{t}[u, \underline{z}]$ and $r\in
\mathbb{N}$, we denote by ${\sf DEG}_r(F, u)$ (resp. ${\sf DEG}_{\leq r}(F, u)$)
the subset of polynomials in $F$ of degree $r$ (resp. at most $r$) in $u$. For a
polynomial $f \in \field{t}[u, \underline{z}]$, we denote by ${\sf coeffs}(f,
u)$ its coefficients when $f$ is seen in $\mathbb{K}'[u]$ with
$\mathbb{K}' := \field{t}[\underline{z}]$.\looseness=-1

\begin{lemma}\label{lemma:necessarycondition_fiber}
  We reuse the notation introduced above. Let~$i$ be greater than~$0$.
  If there is no polynomial in $G_x$
  whose degree w.r.t. $u$ is greater than or equal to $i$, then
  $\cfiber{i}{u}{\scrW}$ is empty.
\end{lemma}

\begin{proof}
  Suppose that~$\cfiber{i}{u}{\scrW}\neq \emptyset$ and
    pick~$\balpha\in\cfiber{i}{u}{\scrW}$. By~\cite[Thm.~2, \S2, Ch.~3,
    p.~130]{CoLiOS07}, there exists~$g\in G_x$ of degree~$\leq i-1$ in~$u$ and
    some $1\leq j \leq \deg_u(g)$ such that: the coefficient of~$u^j$ does not
    vanish at~$\balpha$ and the coefficients of~$u^\ell$ vanish at~$\balpha$, for
    all~$\ell>j$. There thus exist at most $j<i$ solutions to the equation~$g(u,
    \underline{z}=\balpha)=0$. This contradicts the fact that the fiber
    above~$\balpha$ has cardinality at least~$i$, and hence
    that~$\balpha\in\mathcal{F}_i(u, \scrW)$. 
\end{proof}

Let $\ell^{(i)}_u$ be the set of leading coefficients of the
polynomials in~$G_x$ that have degree at least $i$ w.r.t. $u$; we
denote $\ell_u^{(1)}$ by $\ell_u$.

Also for~$g\in G_x$ of positive degree~$i$ in~$u$, we denote by~$\MinHer(g,
  i)$ the set in~$\K(t)[\underline{z}]$ of all~$i\times i$ minors of the Hermite
  quadratic form associated with~$g$ when seen as a polynomial in~$u$.
  By~\cite[Thm.~4.57, p. 130]{ARAG06}, $g(u, \balpha)$ has at least $i$ distinct roots
  when $\balpha$ does not lie in the common zero set of $\MinHer(g, i)$.

Let $\mathcal{S}^{(i)}$ be the set of points $\bbeta~\in\acfield{t}^{k+1}$ such that
the following polynomial constraints are simultaneously satisfied:
\begin{itemize}
\item[\em (a)] all polynomials in $G_u$ vanish at $\bbeta$;
\item[\em (b)] all polynomials in ${\sf coeffs}(f, u)$ for $f\in {\sf DEG}_{\leq
    i-1}(G_x, u)$ vanish at $\bbeta$;
\item[\em (c)]
  at least one polynomial~$\ell\in\ell_x$ does not vanish at~$\bbeta$;
\item[\em (d)] at least one polynomial $\ell\in \ell^{(i)}_u$ does not vanish at~$\bbeta$;
\item[\em (e)] at least one polynomial in $\MinHer(g, i)$
  does not vanish at~$\bbeta$, for some $g \in G_x$ with degree at least $i$
  w.r.t. $u$.
\end{itemize}

Further, we denote by {\em (f)} the disjunction $\vee_{\ell\in \ell_u} \ell\neq
0$. 

To give the intuition, conditions {\em (a), (c)} and {\em (f)}
  characterize the projected sets~$\scrW$ and
  $\pi(W)$. Conditions {\em (b), (d)} and {\em
    (e)} characterize the cardinality of the fiber.
\begin{proposition}\label{prop:elim}
  The projection of~$\mathcal{S}^{(i)}$ onto the $\underline{z}$-coordinate space
  is~$\cfiber{i}{u}{\scrW}$.
\end{proposition}

\begin{proof}
  We start by showing that the set of points which do satisfy {\em (a)}, {\em
    (c)} and {\em (f)} coincides with $\pi(W)$ which, by definition, contains
  $\cfiber{i}{u}{W}$.

  By \cite[Ch. 3, \S2, Thm.~3, p.~156]{CoLiOS07},
  the Zariski closure of $\pi(W)$ is defined by~{\em~(a)}. 
  By Lemma~\ref{lemma:proj_x}, $\scrW$ is defined by the vanishing of all
  polynomials in $G_x$ and the nonvanishing of at least one element in
  $\ell_x$. Note that $\pi(W)$ coincides with the projection of $\scrW$ on the
  $\uline{z}$-coordinate space (which we also denote by $\pi(\scrW)$ by a slight
  abuse of notation). Note that all points of $\scrW$ also satisfy {\em (c)}
  Applying Fact~\ref{fact:elim}\eqref{fact:extension} to $G_x$ and the Zariski
  closure of $\scrW$ shows that~$\pi(\scrW)$ is also contained in the set of
  points which satisfy {\em (f)}. To prove the reverse inclusion, it suffices to
  apply Fact~\ref{fact:elim}\eqref{fact:extension} by lifting the solutions of
  {\em (a)} from the $\uline{z}$-space to points in
  $W$.\\
  \indent It remains to show that the extra conditions {\em (b)}, {\em (d)}
    and {\em (e)} ensure the cardinality condition on the fiber of~$\pi$. First,
    it results from~\cref{lemma:necessarycondition_fiber} that~{\em (b)} is a
    necessary condition to ensure a fiber of cardinality at least~$i$. Next, it
    follows from~\cite[Chap.~3, \S5, Thm.~2, p.~156]{CoLiOS07} and {\em
      (d)} that the cardinality condition on the fiber is reduced to see under
    which condition a univariate polynomial $g(u, \underline{z} = \balpha)$
    (for~$\balpha\in\pi(\scrW)$ and $g\in G_x$) admits at least~$i$ distinct
    roots. Finally, it follows from~\cite[Thm.~4.57]{ARAG06} that~{\em (e)} is
    a necessary and sufficient condition for this.  Also, a subtle
    observation is that~\textit{(f)} guarantees that the vanishing set of the
    possible denominators in the minors in~$\MinHer(g, i)$ is avoided (as they
    are by construction only powers of~$\textsf{LeadingCoefficient}_u(g)$).
\end{proof}

Hence, the algorithm which relies on \Cref{prop:elim} consists in:
\begin{enumerate}
  \item
  Computing a Gröbner basis~$G$ for $\mathcal{J}\subset \field{t}[x, u, \uline{z}]$
  w.r.t. some lexicographic ordering $\succ$ with $x \succ u \succ
  \underline{z}$;
\item Computing the relations that define $\mathcal{S}^{(i)}$ as described above;
  \item
  For each conjunction of constraints defining $\mathcal{S}^{(i)}$: eliminating~$u$
  from $G_x$ and from the defining equations
  and~eliminating the saturation variables introduced to
  handle inequations (still, the inequations should be kept in the output).\\
  \end{enumerate}\vspace{-0.2cm}
  Note that in step ($3$), Gröbner bases can be used to perform the elimination.
  This also has the advantage to determine if there are points which do satisfy
  both the equations and inequations defining~$\mathcal{S}^{(i)}$, hence deciding its
  emptiness.\\
  It should be noted that, in practice, one can avoid to use
  $\field{t}$ as a base field and perform the computations in $\bfield[x, u,
  \uline{z}, t]$ with an elimination ordering where $t$ is the smallest
  variable. Specialization properties of Gröbner bases \cite{Kalkbrener97} show
  that one obtains this way a nonreduced Gröbner basis for $\mathcal{J}$.
  Computing the conditions {\em (a)}-{\em(d)} is straightforward with standard
  computer algebra systems \cite{GaGe13}.

  \smallskip For the application to DDEs, we make the following hypothesis.
\begin{align*}
\label{hyp3} \tag{\bf H3}
\underline{\textsf{Hypothesis 3:}}\;\;\;\; \cfiber{k}{u}{\mathcal{X}}
\text{ and } \mathcal{S}_k(\check{z}_1, \mathcal{X}) \text{ are finite sets.}
\end{align*}
      
\begin{proposition}\label{prop:elim_DDE}
  Let~$P$ be as in~\eqref{eqn:initpoleqn} and~$a\in\mathbb{K}$.
  Assume~\eqref{hyp1}, \eqref{hyp3} and that $P$ is squarefree.
  Denote~$\mathcal{J} :=\langle P, \partial_x P,
  \partial_uP\rangle:(u-a)^\infty\subset \K(t)[x, u, \underline{z}]$.
  If $\mathcal{G}$ is the disjunction of
    conjunction of polynomial
    equations computed by the
    algorithm based on~\Cref{prop:elim}, then there exists one conjunction
    in~$\mathcal{G}$ from which eliminating all variables but~$t$ and~$z_0$
    yields some nonzero~$R\in\K[t, z_0]$ s.t. $R(t, F(t, a))=0$.
\end{proposition}
\begin{proof}
  First, it results from Sard's lemma~\cite[Prop.~B.2, App-8]{SaSc17}
  that if~$P$ is squarefree, then
  $\mathcal{J}\cap\mathbb{K}(t)[\underline{z}]$ is not reduced to~$0$
  (by using the same proof as in~\cite[Lemma~2.3]{BoChNoSa22}, for~$k>1$).
  Note also that by~\cref{prop:elim}, the solution set of~$\mathcal{G}$
  is~$\cfiber{k}{u}{\scrW}$.
  Finally, using~\eqref{hyp1}
  implies that the projection of~$\mathcal{F}_k(u, \mathcal{X})$ onto the
  $z_0$-coordinate space contains the value~$F(t, a)$.
  Hence using~\cite[Thm.~3, \S2, Ch.~3]{CoLiOS07},
  Fact~\ref{fact:elim}\eqref{fact:elimination}
  and eliminating all variables but~$t$
  and~$z_0$ in each
  condition given by~$\mathcal{G}$  either yields, by~\eqref{hyp3},  a nonzero
  polynomial in~$\K[t, z_0]$ 
  or yields the constant polynomial~$1$. In any case,
  one of the conditions in~$\mathcal{G}$
  yields
  by~\eqref{hyp1} and \eqref{hyp3} a nonzero~$R\in\K[t, z_0]$
  annihilating~$F(t, a)$.
\end{proof}

\section{Geometric approach}\label{sec:stickelberger}
\def\scrW{\mathscr{W}}
\def\denom{\mathsf{denom}}

Let~$\scrV\subset\acfield{t}^{k+2}$ be an algebraic set
associated to an
ideal~$\mathcal{J}$ of~$\K(t)[x, u, \underline{z}]$.
For a set of variables (or scalars)
$z_0, \ldots, z_{k-1}$, recall that we use
the notation $\check{z}_1 := z_0, z_2, \ldots, z_{k-1}$. Also,
we consider the canonical inclusion~$j:\mathbb{K}(t)[x, u, \underline{z}]\rightarrow
\K(t, \check{z}_1)[x, u, z_1]$.

In this section, we say that assumption~\hypertarget{hypS}{$\bf{(S)}$} holds
if the following assumptions hold:
\begin{itemize}\label{hypS}
  \item \eqref{hyp:stick} holds (with~$\mathscr{W}$ replaced by~$\mathscr{V}$),
\item  the image of $\scrV$ by~$\pi_{\check{z}_1}$ is
  Zariski dense ($\overline{\pi_{\check{z}_1}(\scrV)}
        =
        \acfield{t}^{k-1}$),
\item $\mathcal{J}$ has dimension~$k-1$ in~$\K(t)[x,u, \underline{z}]$,
\item $\mathcal{J}_{z_1}:=\langle j(\mathcal{J}) \rangle$ has dimension~$0$
  in~$\K(t, \check{z}_1)[x, u, z_1]$.
\end{itemize}
\indent \;\;\;
Recall that using~\cref{lemma:SW_constructible}, the set~$\stickV$ is constructible.
 In this section, we design an algorithm with the following
specification: it takes as input a finite set of polynomials
of~$\K(t)[x, u, \underline{z}]$ generating a radical ideal $\mathcal{J}$
satisfying assumption~~\hyperlink{hypS}{$\bf{(S)}$} and such
that~$\mathcal{J}\cap\K(t)[\underline{z}]$ is principal;
it returns, under an additional assumption that will be made explicit later, a
finite set of polynomial constraints whose solution
set is the Zariski closure of~$\stickV$,
for~$\scrV\subset\acfield{t}^{k+2}$ the~zero~set~of~$\mathcal{J}$.

\indent To achieve our aim,
we first determine algebraic relations which induce a zero set
containing~$\stickV$. The ideal~$\mathcal{J}_{z_1}$ having
dimension~$0$, the quotient ring $\mathbb{K}(t, \check{z}_1)[x, u,
z_1]/\mathcal{J}_{z_1}$ defines a $\mathbb{K}(t, \check{z}_1)$-vector space of
finite dimension. We introduce the multiplication map $m_{z_1}: f\mapsto f\cdot
z_1$ which maps $\mathbb{K}(t, \check{z}_1)[x, u, z_1]/ \mathcal{J}_{z_1} $ to
itself, and consider its characteristic polynomial $\xi\equiv\xi_{z_1} \in\mathbb{K}(t,
\check{z}_1)[z_1]$. We denote by~$\chi\equiv\chi_{z_1}
\in\mathbb{K}(t)[\underline{z}]$ the numerator
of~$\xi$ w.r.t.~$\check{z}_1$
and by $\denom(\xi)\in\mathbb{K}(t)[\underline{z}]$ its denominator.
As~$\denom(\xi)$ depends only on~$k-1$ variables, $V(\denom(\xi))$
is a priori only defined
in~$\acfield{t}^{k-1}$. Seeing~$\denom(\xi)$ as a polynomial in~$\K(t)[x, u, \underline{z}]$
(resp.~$\K(t)[u, \underline{z}]$,~$\K(t)[\underline{z}]$), its zero set
is~$\acfield{t}^3\times V(\denom(\xi))$ (resp. $\acfield{t}^2\times V(\denom(\xi))$,
$\acfield{t}\times V(\denom(\xi))$).
With a slight abuse of notation, we denote all these sets by~$V(\denom(\xi))$:
the precise definition domain will be
implicitly dependent on the set with which we intersect/take the complement,
etc.\\
\indent Let~$Z :=
(\overline{\mathbb{K}(t)}^{k+2}\setminus V(\denom(\xi))) \cap
\overline{\mathbb{K}(t)}^3\times~\pi_{\check{z}_1}(\scrV)\subset\acfield{t}^{k+2}$.

\begin{lemma}\label{lemma:correspondance_multmap}
  Under~~\hyperlink{hypS}{$\bf{(S)}$} and assuming the radicality of~$\mathcal{J}$,
  the set~$Z$ is a dense subset of~$\acfield{t}^{k+2}$ which satisfies
$V(\mathcal{J}\cap\mathbb{K}(t)[\underline{z}] +
\langle \chi, \partial_{z_1} \chi, \ldots, \partial_{z_1}^{k-1}\chi\rangle)\cap{\pi(Z)}
= {\mathcal{S}_k(\check{z}_1, \scrV\cap Z)}.$
\end{lemma}
\begin{proof}
  The density of~$Z$ follows from 2 facts: (i) the complement of~$V(\denom(\xi))$
  is dense in~$\acfield{t}^{k-1}$; (ii) by~~\hyperlink{hypS}{$\bf{(S)}$},
  the image of $\scrV$ by $\pi_{\check{z}_1}$ is
  Zariski dense.\\
  \indent Now, from the~$0$-dimensionality of~$\mathcal{J}_{z_1}$, the
  eigenvalues of the endomorphism $m_{z_1}$ of
  $\mathbb{K}(t, \check{z}_1)[x, u, z_1]/\mathcal{J}_{z_1}$
  are, by~Fact~\ref{fact:elim}\eqref{fact:eigenvalue},
  the $z_1$-coordinates of the zero set~$\scrV_{z_1}\subset\overline{\K({t, \check{z}_1})}^3$
  associated to~$\mathcal{J}_{z_1}$. Also note that we can specialize
  $\underline{z}$ in~$\xi$ to any point of $\pi(\scrV\cap
  Z)$.\\
  \indent We now prove the direct inclusion. We pick $\balpha=(\alpha_0,
  \ldots, \alpha_{k-1})\in V(\mathcal{J}\cap\mathbb{K}(t)[\underline{z}] +
  \langle \chi, \partial_{z_1} \chi, \ldots,
  \partial_{z_1}^{k-1}\chi\rangle)\cap\pi(Z)$, and we consider the specialized
  ideal~$\mathcal{J}_\balpha$ obtained by specialization of
  $\mathcal{J}$ to~$\check{z}_1=\check{\balpha}$. As
  $(\check{\balpha})\in
  \pi_{\check{z}_1}(Z)\subset\pi_{\check{z}_1}(\scrV)$, it follows~$V(\mathcal{J}_\balpha)$
  is not empty and hence that~$\mathcal{J}_\balpha$
  has dimension~$0$. Moreover,~$V(\mathcal{J}_{\bm{\alpha}})$
  corresponds to the zero set of the specialized ideal~$\mathcal{J}_{z_1, \balpha}$
  obtained after specialization of~$\mathcal{J}_{z_1}$ to
  $\check{z}_1=\check{\balpha}$. As
  $\balpha\in V(\chi, \partial_{z_1} \chi, \ldots, \partial_{z_1}^{k-1}\chi)$, we
  have that $z_1 = \alpha_1$ is a root of multiplicity
  at least~$k$ of the polynomial $\chi(z_1, \check{z}_1 = \check{\balpha})$.
 By the
 radicality of~$\mathcal{J}$, the use of~Fact~\ref{fact:elim}\eqref{fact:eigenvalue}
 in the
  $0$-dimensional ideal $\mathcal{J}_{z_1, \balpha}$, the correspondence
 $V(\mathcal{J}_\balpha) = V(\mathcal{J}_{z_1, \balpha})$, the fact that
 $\balpha\notin V(\textsf{denom}(\xi))$ and~$\balpha\in
 V(\langle \chi, \partial_{z_1} \chi, \ldots,
  \partial_{z_1}^{k-1}\chi\rangle)$, there exist at least
  $k$ points in~$\scrV\cap(\cap_{i=0, i \neq 1}^{k-1} \{z_i=\alpha_i\})$.
  Hence~$\alpha\in\mathcal{S}_k(\check{z}_1, \scrV\cap Z)$.
    
  \indent We now prove the reverse inclusion. Let~$\balpha=(\alpha_0, \ldots,
  \alpha_{k-1}) \in\mathcal{S}_{k}(\check{z}_1, \scrV \cap Z)$. By definition
  of~$\mathcal{S}_{k}(\check{z}_1, \scrV\cap Z)$, we have that $\balpha\in\pi(\scrV\cap Z)$.
  As we have~$\pi(\scrV\cap Z)\subset\pi(\scrV)\cap\pi(Z)$, it follows that
  $\balpha\in\pi(Z)$. It remains to show that~$\balpha\in
  V(\mathcal{J}\cap\mathbb{K}(t)[\underline{z}] + \langle \chi, \ldots,
  \partial_{z_1}^{k-1}\chi\rangle)$. By Fact~\ref{fact:elim}\eqref{fact:elimination},
  we have that~$\balpha\in
  V(\mathcal{J}\cap\mathbb{K}(t)[\underline{z}])$. Now,
  as~$\pi_{\check{z}_1}(\balpha)\notin V(\denom(\xi))$, we can consider for $0\leq i
  \leq k-1$ the well-defined polynomials $\partial_{z_1}^i\chi (z_1, \check{z}_1 =
  \check{\balpha})$. By the set equality $V(\mathcal{J}_\balpha) = V(\mathcal{J}_{z_1, \balpha})$,
  the fact that there exist
  at least $k$ points in~$\scrV\cap(\cap_{i=0, i\neq1}^{k-1}\{z_i=\alpha_i\})$
  translate, by the application of Fact~\ref{fact:elim}\eqref{fact:eigenvalue}
  to the
  well-defined $0$-dimensional ideal~$\mathcal{J}_{z_1, \balpha}$,
  to the vanishings
  ~$\partial_{z_1}^i\chi(\underline{z}=\balpha)=0$, for all~$0\leq i \leq k-1$.
\end{proof}

In all the combinatorial examples that we considered so far,
taking $\mathcal{J} := \langle P, \partial_x P, \partial_u P\rangle:(u-a)^\infty$
for~$P$ as in~\eqref{eqn:initpoleqn}
always led to~$\mathcal{S}_k(\check{z}_1, \overline{\mathcal{X}}\cap Z) =
\mathcal{S}_k(\check{z}_1, \overline{\mathcal{X}})$.
Hence in order to simplify things and to avoid introducing additional technicalities, 
we assume in the rest of this section that
\begin{align*}
  \label{hyp:ua} \tag{${\mathcal{Z}}$}
  \text{ $\mathcal{S}_k(\check{z}_1, \scrV\cap Z) = \stickV$.}
\end{align*}
From an application viewpoint, working under this new generic assumption is
(as mentioned above) harmless. 

  \begin{lemma}\label{lemma:characterization_stickelberger}
    Assuming~\eqref{hyp:ua},~\hyperlink{hypS}{$\bf{(S)}$},
    and the radicality of~$\mathcal{J}$,
    the Zariski closure of~$\stickV$ is the zero set
    of the saturated ideal
    $(\mathcal{J}\cap\mathbb{K}(t)[\underline{z}]+
    \langle \chi, \ldots, \partial_{z_1}^{k-1}\chi\rangle):\operatorname{denom}(\xi)^\infty$.
  \end{lemma}
  \begin{proof}
    We define $V_1:=V(\mathcal{J}\cap\mathbb{K}(t)[\underline{z}]+
    \langle \chi, \ldots, \partial_{z_1}^{k-1}\chi\rangle)$.
    Using~\eqref{hyp:ua} we can
    replace~$\stickV$ by~$\mathcal{S}_k(\check{z}_1, \scrV\cap Z)$.
    Also, \cref{lemma:correspondance_multmap} implies that
    the Zariski closure of~$\stickV$ is equal to the Zariski closure
    of~$V_1\cap\pi(Z)$. Denote $W:=\overline{\mathbb{K}(t)}\times\pi_{\check{z}_1}(\scrV)\subset
    \overline{\mathbb{K}(t)}\times\overline{\mathbb{K}(t)}^{k-1}$
    (which is dense
    in~$\overline{\mathbb{K}(t)}^{k}$ as~the image of $\scrV$ by
    $\pi_{\check{z}_1}$ is assumed to be Zariski dense).
    By the definition of~$Z$, we
    have~$\overline{V_1\cap\pi(Z)} =
    \overline{V_1\cap\big(W\setminus \pi(V(d(\xi)))\big)}$. Using the density
    property of~$W$, we thus have that $\overline{V_1\cap\pi(Z)} =
    \overline{V_1\setminus \pi(V(d(\xi)))}$. The results hence follows
    as the zero set of~$(\mathcal{J}\cap\mathbb{K}(t)[\underline{z}]+
    \langle \chi, \ldots, \partial_{z_1}^{k-1}\chi\rangle):d(\xi)^\infty$ is
    $\overline{V_1\setminus \pi(V(d(\xi)))}$.
\end{proof}

Applying~\cref{lemma:correspondance_multmap,lemma:characterization_stickelberger}
and under~\eqref{hyp:ua}, our aim is hence
to compute~$\mathcal{J}\cap\mathbb{K}(t)[\underline{z}] +
\langle \chi, \partial_{z_1} \chi, \ldots, \partial_{z_1}^{k-1}\chi\rangle$.
Regarding complexities estimates,
the use of parametric geometric resolution tools (\cite{Schost03}) allows to prove
the following result.
\begin{lemma}\label{lemma:complexity_stickelberger}
  Assume that~$\mathcal{J}$ is radical
  of degree~$D$, of dimension~$k-1$ and
  that: $\mathcal{J}\cap\mathbb{K}(t)[\underline{z}]$ is principal and
  $\mathcal{J}$
  induces an ideal~$\mathcal{J}_{z_1}\subset\mathbb{K}(t, \underline{z})[x,u, z_1]$
  of dimension~$0$.
Suppose given a straight-line program of length~$L$ evaluating the polynomials
defining~$\mathcal{J}$. 
Then computing a generator
of~$\mathcal{J}\cap\mathbb{K}[t, \underline{z}]$
and the polynomials $\chi, \partial_{z_1}\chi, \ldots,
\partial_{z_1}^{k-1}\chi$ can be done
in~$\tilde{O}(D^{k+1}\cdot 8^k\cdot(L+k^2) + k\cdot D^{2(k+1)})$ ops. in~$\mathbb{K}$.
\end{lemma}
\begin{proof}
We first denote by~$\mathcal{J}_{z_1}$ the induced ideal of dimension~$0$
in~$\mathbb{K}(t, \check{z}_1)[x,u, z_1]$.
Using the radicality and the dimension assumptions, it is possible
to compute a parametric geometric resolution of the zero set of~$\mathcal{J}_{z_1}$.
Applying the algorithm underlying~\cite[Theorem~2]{Schost03},
we compute $A, B\in\mathbb{K}(t, \hat{\underline{z}})[\lambda]$
which give a parametrization $z_1= A(\lambda)/\partial_\lambda B(\lambda)$ over the field
extension defined by $B(\lambda)=0$. Using~\cite[Theorem~2]{Schost03}, this can be done
in~$\tilde{O}(D^{k+1}\cdot 8^k\cdot (L+k^2))$ ops.
in~$\mathbb{K}$. Applying Fact~\ref{fact:elim}\eqref{fact:Stickelberger},
$\chi$ is equal to
the resultant in~$\lambda$ of the numerators of $B$
and $z_1\cdot\partial_\lambda B-A$. Using now fast
computation of bivariate resultants~\cite[\S5]{HyNeSc19} and using the
upper bound~$D$~\cite[Theorem~1]{Schost03} on the partial degrees of the coefficients
of both the numerator and denominator of $A$ and~$B$, we obtain that the number of
evaluation-interpolation points needed for $t, \check{z}_1$ is
$O(D^{2k})$. This gives a cost in~$\tilde{O}(D^{2k+1.63})$ ops. in~$\mathbb{K}$ for computing
the numerator of $\chi$, whose partial degrees are bounded by~$D^2$. Finally,
if $L'$ is the length of a straight-line program evaluating~$\chi$, 
Theorem~1 in \cite{BaSt83} allows us to compute
$\partial_{z_1}\chi, \ldots, \partial_{z_1}^{k-1}\chi$ using $\tilde{O}(k\cdot L')
\subset \tilde{O}(k\cdot D^{2(k+1)})$
ops. in~$\mathbb{K}$. This yields a final complexity
in~$\tilde{O}(D^{k+1}\cdot 8^k\cdot(L+k^2) + k\cdot D^{2(k+1)})$ ops. in~$\mathbb{K}$.
As~$\mathcal{J}$ is assumed radical,
a generator of~$\mathcal{J}\cap\mathbb{K}[t, \underline{z}]$ is given
by the squarefree part of~$\chi$. 
The cost of this squarefree part computation is negligible and absorbed in
the above complexity.
\end{proof}

\begin{remark}\label{rmk:key_geom_sat}
  {(${\mathbf i}$)} When the numerator of the
  characteristic polynomial of $m_{z_1}$ generates
  $\mathcal{J}\cap\mathbb{K}[t, \underline{z}]$,
  the complexity of \cref{lemma:complexity_stickelberger} drops
to~$\tilde{O}(D^{k+1}\cdot 8^k\cdot(L+k^2) + k\cdot D^{k+1})\subset
\tilde{O}(D^{k+1}\cdot 8^k\cdot(L+k^2))$. 

\smallskip {(${\mathbf {ii}}$)}
\cref{lemma:complexity_stickelberger}
allows with the same complexity to compute the characteristic polynomial~$\xi_u$ of
the multiplication map $m_u$. Denoting by~$\chi_u$ the numerator of~$\xi_u$,
the refinement of $\stickV$ consisting in
counting only the distinct solutions w.r.t.~$u$
is equivalent to considering the polynomial
$m\cdot \textsf{disc}_u(\chi_u)-1$ (and then eliminating~$m$).
Another useful practical refinement consists in adding to the polynomials
defining~$\overline{\stickV}$
all the polynomials defining~$\overline{\mathcal{S}_k(\check{z_i}, \scrV)}$,
for all~$0\leq i \leq k-1$ (by a slight adaptation of the present section).
\end{remark}

\cref{lemma:complexity_stickelberger,lemma:characterization_stickelberger}
yield an algorithm whose output
characterizes~$\overline{\mathcal{S}_k(\check{z_1}, \scrV)}$, and prove its complexity.
As most of the combinatorial examples we have encountered
so far are stated with~$\mathbb{K}=\mathbb{Q}$, our aim is to 
 take in practice the benefit of fast multi-modular arithmetic. We prefer to reduce
 the computations in~$\mathbb{Q}$ instead of~$\mathbb{Q}(t)$ by using evaluation-interpolation
 on the parameter~$t$.
 In practice, the underlying algorithm
 makes use of the specialization properties of
Gr\"obner bases~\cite[Prop.~1, p.~308]{CoLiOS07}.

\begin{proposition}\label{prop:stick_DDE}
Let~$P\in\mathbb{K}[x, \underline{z}, t, u]$
be as in~\eqref{eqn:initpoleqn}, $a\in\mathbb{K}$ and
define $\mathcal{J}:=\langle P, \partial_xP, \partial_uP\rangle :(u-a)^\infty
\subset\mathbb{K}(t)[x,u, \underline{z}]$.
We assume that:~\eqref{hyp1}, \eqref{hyp3}, ~\eqref{hyp:ua} and~\hyperlink{hypS}{$\bf{(S)}$}
hold and that
the input assumptions of~\cref{lemma:complexity_stickelberger}
are satisfied by the ideal
$\mathcal{J}$.
Then
any nonzero $R\in(\mathcal{J}\cap\mathbb{K}(t)[\underline{z}] +
\langle \chi, \partial_{z_1} \chi, \ldots, \partial_{z_1}^{k-1}\chi\rangle)\cap
\mathbb{K}[t, z_0]$ satisfies~$R(t, F(t, a))=0$.
\end{proposition}
\begin{proof}
  Using the algorithm underlying the proof
  of~\cref{lemma:complexity_stickelberger},
  we compute a generator of
  $\mathcal{J}\cap\mathbb{K}(t)[\underline{z}]$ and 
  $\chi, \partial_{z_1} \chi, \ldots, \partial_{z_1}^{k-1}\chi$.
  Now using~\cref{lemma:characterization_stickelberger},
  the zero set of~$(\mathcal{J}\cap\mathbb{K}(t)[\underline{z}] +
  \langle \chi, \partial_{z_1} \chi, \ldots, \partial_{z_1}^{k-1}\chi\rangle)$
  is~$\mathcal{S}_k(\check{z}_1, \overline{\mathcal{X}})$.
  Hence applying \cite[Thm.~$3$, \S2, Ch.~3]{CoLiOS07},
  the zero set of $(\mathcal{J}\cap\mathbb{K}(t)[\underline{z}] +
\langle \chi, \partial_{z_1} \chi, \ldots, \partial_{z_1}^{k-1}\chi\rangle)\cap
\mathbb{K}[t, z_0]$ is the Zariski closure of the
projection of $\mathcal{S}_k(\check{z}_1, \overline{\mathcal{X}})$ onto the $z_0$-coordinate space, which by~\eqref{hyp1} contains~$F(t, a)$. Hence
as by~\eqref{hyp3} the latter elimination ideal is not reduced to $0$,
any element of it annihilates~$F(t, a)$.
\end{proof}

\section{Conclusion and perspectives}\label{sec:conclusion}
Extensive practical experiments on DDEs of
  type~\eqref{eqn:initfunceqn}--\eqref{eqn:initpoleqn} 
defined by dense polynomials $f$ and $Q$ show that
the growth order estimate~$\delta^{3k}$ 
for the algebraicity degree of $F(t,a)$
in Prop.~\ref{prop:bounds_under_H1}
is very likely to be sharp in the worst case,
and actually reached in the ``generic'' case.
For instance, when~$k=1$ we
  observe that, on random examples, the minimal
  polynomial~$M\in\mathbb{Q}[t, z_0]$ of~$F(t, a)$
  satisfies~$\deg_{z_0} M = \delta(\delta^2-\delta+1)$ 
  and~$\deg_t M = 2\delta^3+\delta^2-3\delta+2$. 
  For~$k=2$, we managed to compute the degrees in~$z_0$
  for~$\delta = 4, 7, 10, 13, 16$ (corresponding to~$\deg(Q)
  = \deg(f) \in \{ 1, 2, 3, 4, 5 \}$)
  and obtained successively $\deg_{z_0}M = 1, 38, 870, 5824$ and $24235$.
  This makes us very confident that
  the asymptotic growth of $\deg_{z_0}M$ is of
  order~$\delta^{3\cdot 2} = \delta^{6}$. However, we do not have a 
  proof that $\delta^{3k}$ indeed matches the right 
  order of magnitude of $\deg_{z_0}M$ and of $\deg_{t}M$.
\indent Since we believe that the output $R$ of our algorithms has (generically) arithmetic size 
   $A = \deg_{z_0}R \cdot \deg_{t}R = \delta^{3k} \cdot \delta^{3k} = \delta^{6k}$,
  exponentiality in $k$ is unavoidable in the complexity estimates: any
  algorithm for computing   
  $R$ would need at least $\delta^{6k}$ ops. From this perspective, the
  estimate in 
  Prop.~\ref{prop:complexity_under_H1} is quite good, since it is $O(A^{4/3})$.

Our algorithms are not only fast in theory, but also efficient
in practice.
Moreover, they allow to
solve nontrivial combinatorial applications, as showed by the experimental
results in \cref{sec:experiments}. 
Our implementations yielded practical improvements for a large majority of
them, and
allowed us to solve
one (5-constellations)
on which the state-of-the-art methods could not terminate. \\
For future works, we wish to develop 
complete implementations of the
algorithms that we have introduced in the present paper, and to make them available
for the combinatorics community.
A different, more theoretical direction, is to
pursue the geometrical investigation analysis of the problem of
computing exceptional fibers initiated in~\cref{sec:elimination,sec:stickelberger}.

\section{Experiments}\label{sec:experiments}

  \paragraph{Aim} We first report on practical variants of the hybrid
  guess-and-prove method \textsf{(hgp)} and then provide and analyze tables
  of our implementations of
  \cref{sec:direct,sec:guessandprove,sec:elimination,sec:stickelberger}.
   The benchmark DDEs on which we test our various implementations 
  have combinatorial origins and the literature qualifies their
  resolution as a highly nontrivial problem. 
  More precisely, we consider solving:
  Eq.~(4.22) in \cite{Bernardi08}
  (``near-triangulations''), 
  Prop.~12 in \cite{BMJ06} (``$m$-constellations'', $m\in\{4, 5\}$) 
  and
  Eq.~(3) in \cite{BMFP12} (``3-Tamari lattices'').\\
  \indent Recall that the input of the algorithms
  in~\cref{sec:direct,sec:guessandprove,sec:elimination,sec:stickelberger} 
  consists of 
  a polynomial
  $P$ as in~\eqref{eqn:initpoleqn} and 
  of
  a specialization point $a\in\mathbb{K}$, while
  their output 
  is, under~\eqref{hyp3} and up to eliminating variables,
  a nonzero  polynomial~$R\in\mathbb{K}[t, z_0]$
  such that~$R(t, F(t, a)) = 0$.
  \paragraph{Implementations}
  The DDEs we consider being defined over~$\K=\Q$, we use multi-modular arithmetic
  and CRT (Chinese Remainder Theorem) for
  \cref{sec:direct,sec:elimination,sec:stickelberger}.
  Also, we reduce the computation from~$\Q(t)$ to
  $\Q$ by performing evaluation-interpolation (``ev.-int.'') on either~$t$ or~$z_0$.
  We incorporate~ $\mathbf{(ii)}$ of Rmk. \ref{rmk:key_geom_sat}, but do
  not use the inequalities describing~$\mathcal{S}^{(k)}$ in~\cref{sec:elimination}. 
  Finally,
  we use standard tools in computer algebra to improve each of our
  implementations.\\
  \indent The practical variant
  of the \textsf{hgp} strategy mentioned above has the following
  motivation.
When one performs (say, for a ``random'' prime~$p$) the computation of a modular image~$R_p$
of~$R$ by using one of~\cref{sec:direct,sec:elimination,sec:stickelberger},
the computation (if it ends) gives access to the partial degrees
of~$R\in\mathbb{Q}[t, z_0]$. If either one modular computation or
the lift over~$\mathbb{Q}$ is too time consuming,
the following variant of \textsf{hgp} exploits the knowledge of those partial degrees:
\begin{itemize}
\item Pick a random prime~$2^{27} < p < 2^{31}$ and
   $\theta\in\{1, \ldots, p-1\}$,
\item
  Compute $R_p(t, \theta) \in\mathbb{F}_p[t]$ and $R_p(\theta, z_0)\in
  \mathbb{F}_p[z_0]$ using one
  of~\cref{sec:direct,sec:elimination,sec:stickelberger};
  set $d_t$ and $d_{z_0}$ their respective degrees,
\item Compute~$F(t, a)\bmod t^{2d_td_{z_0}+1}$,
\item Guess $\tilde{R}\in\mathbb{Q}[t, z_0]$ s.t.
  $\tilde{R}(t, F(t, a))=O(t^{(d_t+1)(d_{z_0}+1)-1})$,
\item Check that~$\tilde{R}(t, F(t, a)) = O(t^{d_t\cdot \deg_{z_0}(\tilde{R}) +
  \deg_t(\tilde{R})\cdot d_{z_0}+1})$.
\end{itemize}
The above algorithm is a simple extension of the one in~\cite[\S2.2.2]{BoChNoSa22},
with the total degree replaced by partial degrees. Also, our implementation generates
terms of~$F(t,a)$ by first computing terms of~$F(t, u)$ and then specializing to~$u=a$.
Any optimization of this step would result in much better timings for
the~\textsf{hgp} strategy.

In our experiments, we consider the following data:
\begin{itemize}
  \item $\mathbf{S}$: sections (and hence algorithms) used,
\item $\mathbf{\#\mathcal{P}}$: number of primes used for the CRT,
  \item $\mathbf{Z}\in\{t, z_0\}$ the variable on which we perform ev.-int.,
  \item \textbf{$\mathbf{\#}$pts}: number of ev.-int. points needed in~$Z$,
  \item $\mathbf{d}_{\textbf{cp}}$: critical pairs of maximal degree in GB computations,
    \item $\mathbf{d}_\mathbf{M}$: Macaulay matrix maximal size in GB computations (F$4$),
      \item $\mathbf{d}_{\mathcal{I}_\infty}$: degree of the ideal~$\mathcal{I}_\infty$
        in~\cref{sec:direct},
      \item $\mathbf{d}_\chi$: degree of~$\chi_{z_1}\in\mathbb{Q}[t, \underline{z}]$,
  \item $\mathbf{T}$: total timing needed to obtain an output in~$\mathbb{Q}[t, z_0]$,
  \item $\mathbf{d}_\mathbf{Z}$: degree in~$Z$ of output~$R\in\mathbb{Q}[t, z_0]$ s.t.
    $R(t, F(t, a))=0$,
\item $\mathbf{\sigma}$: truncation order in the expansion of~$F(t, a)$,
\item $\mathbf{G}$:
  time spent for guessing an annihilating polynomial in~$\mathbb{Q}[t, z_0]$,
\item $\mathbf{P}$: time spent to prove the guess.
\end{itemize}

\indent \;\;\;The timings are given in seconds (s.), minutes (m.), hours (h.) and days (d.).
The symbol $\infty$ (resp. $-$, $\times$) means
that the computation (resp. the data) did not finish (resp. was not known, is
not defined)
after 5 days.\\
\indent All computations
were conducted using Maple on a computer equipped with
Intel® Xeon® Gold CPU 6246R v4 @ 3.40GHz and 1.5TB of RAM with a single thread.
All Gr\"obner bases computations were performed using the  \textsf{C} library
\href{https://msolve.lip6.fr/}{msolve}~\cite{msolve}, 
and all guessing computations were performed using the
\href{http://perso.ens-lyon.fr/bruno.salvy/software/the-gfun-package/}{gfun} Maple
package~\cite{Gfun}.

We obtain the following tables:

\tiny
\begin{center} {\bf{\footnotesize \cite[Proposition~$\mathbf{4.3}$]{Bernardi08}}, {\bf $k=\mathbf{2}$}}\\
  \begin{tabular}{|c|c|c|c|c|c|c|c|c|c|c|c|}
    \hline
    $\mathbf{S}$ & $\mathbf{\#\mathcal{P}}$ & $\mathbf{Z}$ & \textbf{$\mathbf{\#}$pts}
    & \textbf{$\mathbf{d}_{\textbf{cp}}$} & \textbf{$\mathbf{d}_\mathbf{M}$} &
    \textbf{$\mathbf{d}_{\mathbf{\mathcal{I}}_\infty}$} &
    \textbf{$\mathbf{d}_{\mathbf{\chi}}$}
    & \textcolor{blue}{$\mathbf{T}$} & $\mathbf{d}_\mathbf{t}$ & $\mathbf{d}_{\mathbf{z_0}}$\\
    \hline
    \cref{sec:direct} & $3$ & $t$ & $265$ & $18$ & $4\cdot10^5\times4\cdot10^5$ & $12$ & $\times$ & \textcolor{blue}{$55$s} & $132$ & $6$\\
    \cref{sec:elimination}  & $3$ & $t$ & $265$ & $206$ & $2\cdot10^4\times2\cdot10^4$ &$\times$  & $\times$ & \textcolor{blue}{$1$m$10$s} & $132$ & $6$\\
    \cref{sec:stickelberger} & $\times$ & $\times$ & $\times$ & $\times$ &
    $\times$
    & $\times$ & $469$ & \textcolor{blue}{$30$s} & $1173$ & $33$ \\
    \hline
\end{tabular}
\end{center}

\tiny
\begin{center}
  \begin{tabular}{|c|c|c|c|c|c|c|}
    \hline
    $\mathbf{S}$ & $\mathbf{d}_\mathbf{t}$ & $\mathbf{d}_{\mathbf{z_0}}$ &$\mathbf{\sigma}$&
    $\mathbf{G}$  &
  $\mathbf{P}$ & \textcolor{blue}{$\mathbf{T}$} \\
  \hline
  \cref{sec:guessandprove,sec:stickelberger} &  $132$ \textcolor{blue}{($7$m)}
  & $6$ \textcolor{blue}{($0.4$s)}& $\mod t^{2048}$ \textcolor{blue}{($\infty$)} & $-$ & $-$ &
  $-$
 \\
  \hline
\end{tabular}\\[0.2cm]
\end{center}

\tiny
\begin{center}  {\bf {\footnotesize $\mathbf{4}$-constellations}, {\bf $k=\mathbf{3}$}}\\
  \begin{tabular}{|c|c|c|c|c|c|c|c|c|c|c|c|}
    \hline
    $\mathbf{S}$ & $\mathbf{\#\mathcal{P}}$ & $\mathbf{Z}$ & \textbf{$\mathbf{\#}$pts}
    & \textbf{$\mathbf{d}_{\textbf{cp}}$} & \textbf{$\mathbf{d}_\mathbf{M}$} &
    \textbf{$\mathbf{d}_{\mathbf{\mathcal{I}}_\infty}$} &
    \textbf{$\mathbf{d}_{\mathbf{\chi}}$}
    & \textcolor{blue}{$\mathbf{T}$} & $\mathbf{d}_\mathbf{t}$ & $\mathbf{d}_{\mathbf{z_0}}$\\
    \hline
    \cref{sec:direct} & $3$ & $t$ & $7$ & $9$ & $3\cdot10^4\times3\cdot10^4$ & $42$
    & $\times$ & \textcolor{blue}{$4$m} & $3$ & $7$ \\
    \cref{sec:elimination} & $3$ & $t$ & $7$ & $33$ & $7\cdot 10^3\times 11\cdot 10^3$ & $\times$
    & $\times$ & \textcolor{blue}{$41s$} &  $3$ & $7$ \\
    \cref{sec:stickelberger} & $3$ & $t$ & $7$ &  $51$ &
    $2\cdot 10^5\times 2\cdot 10^5$  & $\times$ & $28$ & \textcolor{blue}{$83s$}
    &  $3$ & $7$  \\
    \hline
\end{tabular}
\end{center}

\tiny
\begin{center}
  \begin{tabular}{|c|c|c|c|c|c|c|}
    \hline
    $\mathbf{S}$ & $\mathbf{d}_\mathbf{t}$ & $\mathbf{d}_{\mathbf{z_0}}$ &$\mathbf{\sigma}$&
    $\mathbf{G}$  &
  $\mathbf{P}$ & \textcolor{blue}{$\mathbf{T}$} \\
  \hline
 \cref{sec:guessandprove,sec:elimination} & $3$  \textcolor{blue}{($1$s)}  & $7$  \textcolor{blue}{($1.5$s)}& $\mod t^{64}$
 \textcolor{blue}{($18$s)} & \textcolor{blue}{$0.03$s} & \textcolor{blue}{$0.001$s} &
  \textcolor{blue}{$20$s}\\
  \hline
\end{tabular}\\[0.2cm]
\end{center}

\tiny
\begin{center}  {\bf {\footnotesize $\mathbf{3}$-Tamari}, {\bf $k=\mathbf{3}$}}\\
  \begin{tabular}{|c|c|c|c|c|c|c|c|c|c|c|c|}
    \hline
    $\mathbf{S}$ & $\mathbf{\#\mathcal{P}}$ & $\mathbf{Z}$ & \textbf{$\mathbf{\#}$pts}
    & \textbf{$\mathbf{d}_{\textbf{cp}}$} & \textbf{$\mathbf{d}_\mathbf{M}$} &
    \textbf{$\mathbf{d}_{\mathbf{\mathcal{I}}_\infty}$} &
    \textbf{$\mathbf{d}_{\mathbf{\chi}}$}
    & \textcolor{blue}{$\mathbf{T}$} & $\mathbf{d}_\mathbf{t}$ & $\mathbf{d}_{\mathbf{z_0}}$\\
    \hline
    \cref{sec:direct} & $4$ & $t$ & $11$ & $11$ & $3\cdot10^5\times3\cdot10^5$ & $96$ & $\times$ & \textcolor{blue}{$2$d$2$h} &  $5$ &  $16$\\
    \cref{sec:elimination} & $4$ & $z_0$ & $33$ &$64$ & $10^4\times10^4$& $\times$ & $\times$&   \textcolor{blue}{$2$m} &  $5$ &  $16$ \\
    \cref{sec:stickelberger} & $4$ & $t$ & $11$ & $52$ & $10^4\times10^4$ & $\times$ & $31$ &   \textcolor{blue}{$5$m$42$s}
    & $5$ &  $16$\\
    \hline
\end{tabular}
\end{center}
\tiny
\begin{center}
  \begin{tabular}{|c|c|c|c|c|c|c|}
    \hline
    $\mathbf{S}$ & $\mathbf{d}_\mathbf{t}$ & $\mathbf{d}_{\mathbf{z_0}}$ &$\mathbf{\sigma}$&
    $\mathbf{G}$  &
  $\mathbf{P}$ & \textcolor{blue}{$\mathbf{T}$} \\
  \hline
 \cref{sec:guessandprove,sec:elimination} & $5$  \textcolor{blue}{($0.2$s)}  & $16$  \textcolor{blue}{($5$s)}& $\mod t^{256}$
 \textcolor{blue}{($1$h$40$m)} & \textcolor{blue}{$1$s} & \textcolor{blue}{$0.2$s} &
  \textcolor{blue}{$1$h$40$m}\\
  \hline
\end{tabular}\\[0.2cm]
\end{center}

\tiny
\begin{center} {\bf {\footnotesize $\mathbf{5}$-constellations}, {\bf$k=\mathbf{4}$}}
  \begin{tabular}{|c|c|c|c|c|c|c|c|c|c|c|c|}
    \hline
    $\mathbf{S}$ & $\mathbf{\#\mathcal{P}}$ & $\mathbf{Z}$ & \textbf{$\mathbf{\#}$pts}
    & \textbf{$\mathbf{d}_{\textbf{cp}}$} & \textbf{$\mathbf{d}_\mathbf{M}$} &
    \textbf{$\mathbf{d}_{\mathbf{\mathcal{I}}_\infty}$} &
    \textbf{$\mathbf{d}_{\mathbf{\chi}}$}
    & \textcolor{blue}{$\mathbf{T}$} & $\mathbf{d}_\mathbf{t}$ & $\mathbf{d}_{\mathbf{z_0}}$\\
    \hline
    \cref{sec:direct} & $-$ & $-$ & $-$& $-$& $-$& $-$&$\times$ & \textcolor{blue}{$\infty$} & $-$ & $-$\\
    \cref{sec:elimination}  & $-$ & $t$ & $53$ & $70$ & $2\cdot10^6\times2\cdot10^6$
    & $\times$ & $\times$ &
    \textcolor{blue}{$\infty$} & $26$ & $53$\\
    \cref{sec:stickelberger} & $-$ & $-$&$-$ &$-$&$-$& $\times$ & $-$& \textcolor{blue}{$\infty$} & $-$& $-$\\
    \hline
\end{tabular}
\end{center}

\tiny
\begin{center}
  \begin{tabular}{|c|c|c|c|c|c|c|}
    \hline
    $\mathbf{S}$ & $\mathbf{d}_\mathbf{t}$ & $\mathbf{d}_{\mathbf{z_0}}$ &$\mathbf{\sigma}$&
    $\mathbf{G}$  &
  $\mathbf{P}$ & \textcolor{blue}{$\mathbf{T}$} \\
  \hline
 \cref{sec:guessandprove,sec:elimination} & $26$  \textcolor{blue}{($47$m)}  & $53$  \textcolor{blue}{($45$m)}& $\mod t^{256}$
 \textcolor{blue}{($4$h$35$m)} & \textcolor{blue}{$0.03$s} & \textcolor{blue}{$0.02$s} &
  \textcolor{blue}{$6$h$7$m}\\
  \hline
\end{tabular}\\\vspace{0.2cm}
\end{center}

\normalsize
 \paragraph{Interpretation}

  \indent A first natural observation is that the algorithms
  introduced in 
  \cref{sec:direct,sec:guessandprove,sec:elimination,sec:stickelberger}, 
  as well as their practical variants, are relevant in practice. 
  Moreover, for all the examples, there is always one of the
  new methods which is more efficient in terms of timings than the
  state-of-the-art~(\cref{sec:direct}). According to the tables, there is generally no unique
  method that is always better than the others. On the contrary, the experiments
  show that all the new methods can be useful in practice, depending on the DDEs 
  under study
  (and hence on the properties of the associated zero sets).
  \\

  \indent We now explain the tables related to~$5$-constellations. Note first that
  neither~\cref{sec:direct} nor~\cref{sec:stickelberger} allow to compute any single
  specialization (at~$t$ or~$z_0$ specialized) of a modular image of~$R$.
  Now applying~\cref{sec:elimination}, we manage to compute two specializations
  (first in~$t$, then in~$z_0$) of~$R_p$, for some ``random'' prime $p$. This hence
  gives all the relevant data of the line except $\mathbf{\#\mathcal{P}}$ and~$\mathbf{T}$.
  Those two specializations take respectively~$45$m. and~$47$m. The degrees obtained being
  $d_{z_0}=53$ and~$d_t=26$,
  we would need approximately $53\cdot45$min$\, =39$h for each modular computation.
  Estimating the number of such modular computations to be~$5$ (which is very likely
  optimistic), we would hence need at least~$8$ days. Instead of this, we
  use the practical variant of~\cref{sec:guessandprove} mentioned in the current section.
  As the degree of the guessed
  polynomial is low ($\deg_t=2$ and $\deg_{z_0} = 5)$,
  it allows us to compute $256$ terms of the series~$F(t,1)$
  (here $a=1$), and to check the guess with the geometric bounds~$d_t, d_{z_0}$ obtained
  previously.

\balance
\bigskip

\section*{Acknowledgements}
We thank the reviewers for their very helpful comments and suggestions to improve our paper.
The three authors are supported 
by the French grant \textsc{DeRerumNatura} (ANR-19-CE40-0018)
and by the French–Austrian project \textsc{EAGLES} (ANR-22-CE91-0007 \& FWF I6130-N).
The last author is supported by the joint ANR-FWF ANR-19-CE48-0015 \textsc{ECARP} project, 
and the European Union's Horizon 2020 research and innovation programme under the
  Marie Sk\l{}odowska-Curie grant agreement N. 813211 (POEMA).

\end{document}